\documentclass[draft]{article}
\usepackage{amsfonts,amsmath,amssymb,color,graphicx,latexsym}
\definecolor{monblu}{cmyk}{0.80,0.00,0.00,0.70}
\definecolor{monmag}{cmyk}{0.00,1.00,0.00,0.00}
\newtheorem{lemma}{Lemma}
\newtheorem{proposition}{Proposition}
\newenvironment{proof}{{\bf Proof:}}{~$\dashv$\\}
\def\N{\mathbb{N}}
\def\axiomN{\mathbf{N}}
\def\Axiom{\mathbf{A}}
\def\Rule{\mathbf{R}}
\def\LIK{\mathbf{LIK}}
\def\Fo{\mathbf{Fo}}
\def\card{\mathtt{Card}}
\def\nPlt{\mathtt{nlt}}
\def\At{\mathbf{At}}
\def\K{\mathbf{K}}
\def\S{\mathbf{S}}
\def\IPL{\mathbf{IPL}}
\def\hauteur{\mathtt{h}}
\def\degre{\mathtt{deg}}
\begin{document}
\title{Intuitionistic modal logic $\LIK4$ is decidable}
\author{Philippe Balbiani$^{a}$\footnote{Corresponding author.
Email address: philippe.balbiani@irit.fr.}
\hspace{0.1cm}
\c{C}i\u{g}dem Gencer$^{a,b}$\footnote{Email address: cigdem.gencer@irit.fr.}
\hspace{0.1cm}
Tinko Tinchev$^{c}$\footnote{Email address: tinko@fmi.uni-sofia.bg.}}
\date{$^{a}$Toulouse Institute of Computer Science Research
\\
CNRS--INPT--UT, Toulouse, France
\\
$^{b}$Faculty of Arts and Sciences
\\
Istanbul Ayd\i n University, Istanbul, Turkey
\\
$^{c}$Faculty of Mathematics and Informatics
\\
Sofia University St. Kliment Ohridski, Sofia, Bulgaria}
\maketitle
\begin{abstract}
In this note, we prove that intuitionistic modal logic $\LIK4$ is decidable.
\end{abstract}
Decidability.
Tableaux-based procedures.
Finite frame property.
\section{Introduction}\label{section:introduction}
In this note, we prove that intuitionistic modal logic $\LIK4$ is decidable.
\\
\\
In Sections~\ref{section:trees:and:bips}, we introduce labelled trees, i.e. the relational structures that will be our main tools in our proof of the decidability of $\LIK4$.
In Section~\ref{section:closed:sets:of:formulas}, we introduce the syntax of $\LIK4$.
In Section~\ref{section:semantics}, we introduce the relational semantics of $\LIK4$.
In Section~\ref{section:axiomatization}, we present a Hilbert-style axiomatization of $\LIK4$.
In Section~\ref{section:canonical:model}, we present the canonical model construction of $\LIK4$.
In Sections~\ref{section:maximal:worlds}--\ref{section:finite:model:property}, we prove that intuitionistic modal logic $\LIK4$ is decidable.
\\
\\
This paper includes the proofs of several of our results.
Some of these proofs are relatively simple and we have included them here just for the sake of the completeness.
\section{Some notations}
For all $n{\in}\N$, $(n)$ denotes $\{a{\in}\N$: $1{\leq}a{\leq}n\}$.
\\
\\
For all sets $U,V$, when we write ``$U{\subset}V$'', we mean ``$U{\subseteq}V$ and $V{\not\subseteq}U$''.
\\
\\
For all sets $\Sigma$, $\wp(\Sigma)$ denotes the {\em powerset of $\Sigma$.}
\\
\\
For all sets $\Sigma$, $\card(\Sigma)$ denotes the {\em cardinal of $\Sigma$.}
\\
\\
For all sets $U,V$ and for all bijective functions $f$: $U\longrightarrow V$, $f^{{-}1}$ denotes the {\em inverse function of $f$,} 	that is to say the bijective function $g$: $V\longrightarrow U$ such that for all $x{\in}U$, $g(f(x)){=}x$ and for all $y{\in}V$, $f(g(y)){=}y$.
\\
\\
For all sets $W$ and for all binary relations $R,S$ on $W$, ${R}{\circ}{S}$ denotes the {\em composition of $R$ and $S$,} that is to say the binary relation $T$ on $W$ such that for all $s,t{\in}W$, $s{T}t$ if and only if there exists $u{\in}W$ such that $s{R}u$ and $u{S}t$.
\\
\\
For all sets $W$ and for all binary relations $R$ on $W$, $R^{+}$ denotes the least transitive binary relation on $W$ containing $R$ and $R^{\star}$ denotes the least reflexive and transitive binary relation on $W$ containing $R$.
\\
\\
For all sets $W$, for all binary relations $R$ on $W$ and for all $s{\in}W$, let ${R}(s){=}\{t{\in}W$: $s{R}t\}$.
\\
\\
For all sets $W$, a {\em preorder on $W$}\/ is a reflexive and transitive binary relation on $W$.
\\
\\
For all sets $W$ and for all preorders $\leq$ on $W$, $\geq$ denotes the preorder on $W$ such that for all $s,t{\in}W$, $s{\geq}t$ if and only if $t{\leq}s$.
\\
\\
For all sets $W$, for all preorders $\leq$ on $W$ and for all $s,t{\in}W$, when we write ``$s{<}t$'', we mean ``$s{\leq}t$ and $t{\not\leq}s$''.
\\
\\
For all sets $W$ and for all preorders $\leq$ on $W$, a subset $U$ of $W$ is {\em $\leq$-closed}\/ if for all $s,t{\in}W$, if $s{\in}U$ and $s{\leq}t$ then $t{\in}U$.
\\
\\
For all sets $W$, a {\em partial order on $W$}\/ is an antisymmetric preorder on $W$.
\\
\\
``$\IPL$'' stands for ``Intuitionistic Propositional Logic''.
\section{Labelled trees}\label{section:trees:and:bips}
Let $P$ be a finite set and $\leq$ be a partial order on $P$.
\\
\\
A {\em $P$-labelled tree}\/ is a triple $(N,{E},\lambda)$ where the finite set $N$, the binary relation $E$ on $N$ and the function $\lambda$: $N\longrightarrow P$ are such that
\begin{itemize}
\item $(N,{E})$ is a tree,\footnote{We assume the reader is at home with the basic definitions concerning trees: roots, paths, etc.}
\item $\lambda$ is {\em $\leq$-monotone,} that is to say for all nodes $i,j$, if $i{E}j$ then $\lambda(i){\leq}\lambda(j)$.
\end{itemize}
For all $P$-labelled trees $(N,{E},\lambda),(N^{\prime},{E^{\prime}},\lambda^{\prime})$, a {\em $P$-embedding of $(N,{E},\lambda)$ into $(N^{\prime},{E^{\prime}},\lambda^{\prime})$}\/ is a function $f$: $N\longrightarrow N^{\prime}$ such that
\begin{itemize}
\item for all $i,j{\in}N$, if $i{E}j$ then $f(i){{E^{\prime}}^{\star}}f(j)$,
\item for all $k{\in}N$, $\lambda(k){=}\lambda^{\prime}(f(k))$.
\end{itemize}
Let $\sim$ be the equivalence relation on the set of all $P$-labelled trees such that for all $P$-labelled trees $(N,{E},\lambda),(N^{\prime},{E^{\prime}},\lambda^{\prime})$, $(N,{E},\lambda){\sim}(N^{\prime},{E^{\prime}},\lambda^{\prime})$ if and only if there exists a $P$-embedding of $(N,{E},\lambda)$ into $(N^{\prime},{E^{\prime}},\lambda^{\prime})$ and there exists a $P$-embedding of $(N^{\prime},{E^{\prime}},\lambda^{\prime})$ into $(N,{E},\lambda)$.
\\
\\
%
%
%At the end of this section, we will show that for all countably indexed families $(N_{m},{E_{m}},\lambda_{m})_{m{\in}\N}$ of $P$-label\-led trees, there exists a strictly increasing function $\theta$: $\N{\longrightarrow}\N$ such that for all $m,n{\in}\N$, $(N_{\theta(m)},{E_{\theta(m)}},\lambda_{\theta(m)}){\sim}(N_{\theta(n)},{E_{\theta(n)}},\lambda_{\theta(n)})$.
%
%
%\\
%\\
%
%
For all $n{\in}\N$, a finitely indexed family $(N_{m},{E_{m}},\lambda_{m})_{m{\in}(n)}$ of $P$-labelled trees is {\em dreary}\/ if there exists $m{\in}(n)$ such that $m{<}n$ and $(N_{m},{E_{m}},\lambda_{m}){\sim}(N_{n},{E_{n}},\lambda_{n})$.
\\
\\
At the end of this section, we will show that for all countably indexed families $(N_{m},{E_{m}},\lambda_{m})_{m{\in}\N}$ of $P$-label\-led trees, there exists $n{\in}\N$ such that the finitely indexed family $(N_{m},{E_{m}},\lambda_{m})_{m{\in}(n)}$ of $P$-labelled trees is dreary.\footnote{See Lemma~\ref{proposition:about:labelled:trees:to:be:dreary}.}
\\
\\
A $P$-labelled tree $(N,{E},\lambda)$ is {\em strict}\/ if
\begin{itemize}
\item for all nodes $i,j$, if $i{E}j$ then $\lambda(i){<}\lambda(j)$.
\end{itemize}
Obviously, for all strict $P$-labelled trees $(N,{E},\lambda)$, the height of the tree $(N,E)$ does not exceed $\card(P)$.
\\
\\
For all $P$-labelled trees $(N,{E},\lambda)$, a {\em duplicate in $(N,{E},\lambda)$}\/ is a couple $(i,j)$ of nodes such that $i{E}j$ and $\lambda(i){=}\lambda(j)$.
\\
\\
Obviously, for all $P$-labelled trees $(N,{E},\lambda)$, if $(N,{E},\lambda)$ is non-strict then $(N,{E},
$\linebreak$
\lambda)$ contains duplicates.
\begin{lemma}\label{lemma:useful:for:elimination:of:duplicates}
Let $(N,{E},\lambda)$ be a $P$-labelled tree and $(i,j)$ be a duplicate in $(N,{E},
$\linebreak$
\lambda)$.
The structure $(N^{\prime},{E^{\prime}},\lambda^{\prime})$ such that
\begin{itemize}
\item $N^{\prime}{=}N{\setminus}\{j\}$,
\item $E^{\prime}$ is the binary relation on $N^{\prime}$ such that for all $k,l{\in}N^{\prime}$, $k{E^{\prime}}l$ if and only if either $k{E}l$, or $k{=}i$ and $j{E}l$,
\item $\lambda^{\prime}$: $N^{\prime}\longrightarrow P$ is the function such that for all $k{\in}N^{\prime}$, $\lambda^{\prime}(k){=}\lambda(k)$,
\end{itemize}
is a $P$-labelled tree.
\end{lemma}
\begin{proof}
Let $k,l{\in}N^{\prime}$.
For the sake of the contradiction, suppose $k{E^{\prime}}l$ and $\lambda^{\prime}(k){\not\subseteq}\lambda^{\prime}(l)$.
Hence, either $k{E}l$, or $k{=}i$ and $j{E}l$.
In the former case, $\lambda(k){\subseteq}\lambda(l)$.
Thus, $\lambda^{\prime}(k){\subseteq}\lambda^{\prime}(l)$: a contradiction.
In the latter case, $\lambda(k){=}\lambda(i)$ and $\lambda(j){\subseteq}\lambda(l)$.
Since $(i,j)$ is a duplicate in $(N,{E},\lambda)$, then $\lambda(i){=}\lambda(j)$.
Since $\lambda(k){=}\lambda(i)$ and $\lambda(j){\subseteq}\lambda(l)$, then $\lambda(k){\subseteq}\lambda(l)$.
Consequently, $\lambda^{\prime}(k){\subseteq}\lambda^{\prime}(l)$: a contradiction.
\medskip
\end{proof}
\begin{lemma}\label{lemma:transforming:labelled:trees:into:strict:labelled:trees}
For all $P$-labelled trees $(N,{E},\lambda)$, if $(N,{E},\lambda)$ is non-strict then there exists a $P$-labelled tree $(N^{\prime},{E^{\prime}},\lambda^{\prime})$ such that $N^{\prime}{\subset}N$ and $(N,{E},\lambda){\sim}(N^{\prime},{E^{\prime}},\lambda^{\prime})$.
\end{lemma}
\begin{proof}
Let $(N,{E},\lambda)$ be a $P$-labelled tree.
Suppose $(N,{E},\lambda)$ is non-strict.
Hence, there exists a duplicate $(i,j)$ in $(N,{E},\lambda)$.
Thus, $i{E}j$ and $\lambda(i){=}\lambda(j)$.
Let $(N^{\prime},{E^{\prime}},\lambda^{\prime})$ be the structure such that
\begin{itemize}
\item $N^{\prime}{=}N{\setminus}\{j\}$,
\item $E^{\prime}$ is the binary relation on $N^{\prime}$ such that for all $k,l{\in}N^{\prime}$, $k{E^{\prime}}l$ if and only if either $k{E}l$, or $k{=}i$ and $j{E}l$,
\item $\lambda^{\prime}$: $N^{\prime}\longrightarrow P$ is the function such that for all $k{\in}N^{\prime}$, $\lambda^{\prime}(k){=}\lambda(k)$.
\end{itemize}
By Lemma~\ref{lemma:useful:for:elimination:of:duplicates}, $(N^{\prime},{E^{\prime}},\lambda^{\prime})$ is a $P$-labelled tree.
\\
\\
Obviously, $N^{\prime}{\subset}N$.
\\
\\
Moreover, the reader may easily verify that the function $f$: $N\longrightarrow N^{\prime}$ such that for all $k{\in}N$, if $k{\not=}j$ then $f(k){=}k$ else $f(k){=}i$ is a $P$-embedding of $(N,{E},\lambda)$ into $(N^{\prime},{E^{\prime}},\lambda^{\prime})$.
\\
\\
In other respect, the reader may easily verify that the function $f^{\prime}$: $N^{\prime}\longrightarrow N$ such that for all $k{\in}N^{\prime}$, $f^{\prime}(k){=}k$ is a $P$-embedding of $(N^{\prime},{E^{\prime}},\lambda^{\prime})$ into $(N,{E},\lambda)$.
\\
\\
Consequently, $(N,{E},\lambda){\sim}(N^{\prime},{E^{\prime}},\lambda^{\prime})$.
\medskip
\end{proof}
%
%
%Lemma~\ref{lemma:about:nice:A:trees:and:relation:equivalence:sim} asserts that every $P$-labelled tree is $\sim$-equivalent to a strict $P$-labelled tree.
%
%
\begin{lemma}\label{lemma:about:nice:A:trees:and:relation:equivalence:sim}
For all $P$-labelled trees $(N,{E},\lambda)$, there exists a strict $P$-labelled tree $(N^{\prime},{E^{\prime}},\lambda^{\prime})$ such that $(N,{E},\lambda){\sim}(N^{\prime},{E^{\prime}},\lambda^{\prime})$.
\end{lemma}
\begin{proof}
By Lemma~\ref{lemma:transforming:labelled:trees:into:strict:labelled:trees}.
\medskip
\end{proof}
For all $P$-labelled trees $(N,{E},\lambda)$ and for all nodes $i$, let $(N^{i},{E^{i}},\lambda^{i})$ be the {\em restriction of $(N,{E},\lambda)$ to ${E^{\star}}(i)$.}
\\
\\
For all $P$-labelled trees $(N,{E},\lambda),(N^{\prime},{E^{\prime}},\lambda^{\prime})$, a {\em $P$-isomorphism from $(N,{E},\lambda)$ onto $(N^{\prime},{E^{\prime}},\lambda^{\prime})$}\/ is a bijective function $f$: $N\longrightarrow N^{\prime}$ such that
\begin{itemize}
\item for all $i,j{\in}N$, $i{E}j$ if and only if $f(i){E^{\prime}}f(j)$,
\item for all $k{\in}N$, $\lambda(k){=}\lambda^{\prime}(f(k))$.
\end{itemize}
Let $\simeq$ be the equivalence relation on the set of all $P$-labelled trees such that for all $P$-labelled trees $(N,{E},\lambda),(N^{\prime},{E^{\prime}},\lambda^{\prime})$, $(N,{E},\lambda){\simeq}(N^{\prime},{E^{\prime}},\lambda^{\prime})$ if and only if there exists a $P$-isomorphism from $(N,{E},\lambda)$ onto $(N^{\prime},{E^{\prime}},\lambda^{\prime})$.
\\
\\
Obviously, $\simeq$ is finer than $\sim$, i.e. for all $P$-labelled trees $(N,{E},\lambda),(N^{\prime},{E^{\prime}},\lambda^{\prime})$, if $(N,{E},\lambda){\simeq}(N^{\prime},{E^{\prime}},\lambda^{\prime})$ then $(N,{E},\lambda){\sim}(N^{\prime},{E^{\prime}},\lambda^{\prime})$.
\\
\\
A strict $P$-labelled tree $(N,{E},\lambda)$ is {\em nice}\/ if
\begin{itemize}
\item for all nodes $i,j,k$, if $k{E}i$, $k{E}j$ and $i{\not=}j$ then $(N^{i},{E^{i}},\lambda^{i}){\not\simeq}(N^{j},{E^{j}},\lambda^{j})$.
\end{itemize}
%
%
%Lemma~\ref{lemma:labelled:trees:restriction:labelled:trees} asserts that given an arbitrary $P$-labelled tree, the restriction operations do not take us out of the set of all $P$-labelled trees.
%
%
%\begin{lemma}\label{lemma:labelled:trees:restriction:labelled:trees}
%
%
%For all $P$-labelled trees $(N,{E},\lambda)$ and for all nodes $i$, $(N^{i},{E^{i}},\lambda^{i})$ is a $P$-labelled tree.
%
%
%\end{lemma}
%
%
%Lemma~\ref{lemma:labelled:trees:strict:restriction:labelled:trees:strict} asserts that given an arbitrary strict $P$-labelled tree, the restriction operations do not take us out of the set of all strict $P$-labelled trees.
%
%
%\begin{lemma}\label{lemma:labelled:trees:strict:restriction:labelled:trees:strict}
%
%
%For all strict $P$-labelled trees $(N,{E},\lambda)$ and for all nodes $i$, $(N^{i},{E^{i}},
%
%
%
%
%
%
%
%
%
%
%
%
%$\linebreak$
%
%
%
%
%
%
%
%
%
%
%
%
%\lambda^{i})$ is a strict $P$-labelled tree.
%
%
%\end{lemma}
%
%
%Lemma~\ref{lemma:labelled:trees:nice:restriction:labelled:trees:nice} asserts that given an arbitrary nice $P$-labelled tree, the restriction operations do not take us out of the set of all nice $P$-labelled trees.
%
%
%\begin{lemma}\label{lemma:labelled:trees:nice:restriction:labelled:trees:nice}
%
%
%For all nice $P$-labelled trees $(N,{E},\lambda)$ and for all nodes $i$, $(N^{i},{E^{i}},\lambda^{i})$ is a nice $P$-labelled tree.
%
%
%\end{lemma}
%
%
For all strict $P$-labelled trees $(N,{E},\lambda)$, a {\em triplicate in $(N,{E},\lambda)$}\/ is a triple $(i,j,k)$ of nodes such that $k{E}i$, $k{E}j$, $i{\not=}j$ and $(N^{i},{E^{i}},\lambda^{i}){\simeq}(N^{j},{E^{j}},\lambda^{j})$.
\\
\\
Obviously, for all strict $P$-labelled trees $(N,{E},\lambda)$, if $(N,{E},\lambda)$ is non-nice then $(N,{E},\lambda)$ contains triplicates.
\begin{lemma}\label{lemma:useful:for:elimination:of:triplicates}
Let $(N,{E},\lambda)$ be a strict $P$-labelled tree and $(i,j,k)$ be a triplicate in $(N,{E},\lambda)$.
The structure $(N^{\prime},{E^{\prime}},\lambda^{\prime})$ such that
\begin{itemize}
\item $N^{\prime}{=}N{\setminus}N^{j}$,
\item $E^{\prime}$ is the binary relation on $N^{\prime}$ such that for all $l,m{\in}N^{\prime}$, $l{E^{\prime}}m$ if and only if $l{E}m$,
\item $\lambda^{\prime}$: $N^{\prime}\longrightarrow P$ is the function such that for all $l{\in}N^{\prime}$, $\lambda^{\prime}(l){=}\lambda(l)$,
\end{itemize}
is a strict $P$-labelled tree.
\end{lemma}
\begin{proof}
Let $l,m{\in}N^{\prime}$.
For the sake of the contradiction, suppose $l{E^{\prime}}m$ and $\lambda^{\prime}(l){\not\subset}\lambda^{\prime}(m)$.
Hence, $l{E}m$.
Thus, $\lambda(l){\subset}\lambda(m)$.
Consequently, $\lambda^{\prime}(l){\subset}\lambda^{\prime}(m)$: a contradiction.
\medskip
\end{proof}
\begin{lemma}\label{lemma:transforming:strict:labelled:trees:into:nice:labelled:trees}
For all strict $P$-labelled trees $(N,{E},\lambda)$, if $(N,{E},\lambda)$ is non-nice then there exists a strict $P$-labelled tree $(N^{\prime},{E^{\prime}},\lambda^{\prime})$ such that $N^{\prime}{\subset}N$ and $(N,{E},\lambda){\sim}
$\linebreak$
(N^{\prime},{E^{\prime}},\lambda^{\prime})$.
\end{lemma}
\begin{proof}
Let $(N,{E},\lambda)$ be a strict $P$-labelled tree.
Suppose $(N,{E},\lambda)$ is non-nice.
Hence, there exists a triplicate $(i,j,k)$ in $(N,{E},\lambda)$.
Thus, $k{E}i$, $k{E}j$, $i{\not=}j$ and $(N^{i},{E^{i}},\lambda^{i}){\simeq}(N^{j},{E^{j}},\lambda^{j})$.
Consequently, there exists a $P$-isomorphism $f_{ij}$ from $(N^{i},{E^{i}},\lambda^{i})$ onto $(N^{j},{E^{j}},\lambda^{j})$.
Let $(N^{\prime},{E^{\prime}},\lambda^{\prime})$ be the structure such that
\begin{itemize}
\item $N^{\prime}{=}N{\setminus}N^{j}$,
\item $E^{\prime}$ is the binary relation on $N^{\prime}$ such that for all $l,m{\in}N^{\prime}$, $l{E^{\prime}}m$ if and only if $l{E}m$,
\item $\lambda^{\prime}$: $N^{\prime}\longrightarrow P$ is the function such that for all $l{\in}N^{\prime}$, $\lambda^{\prime}(l){=}\lambda(l)$.
\end{itemize}
By Lemma~\ref{lemma:useful:for:elimination:of:triplicates}, $(N^{\prime},{E^{\prime}},\lambda^{\prime})$ is a strict $P$-labelled tree.
\\
\\
Obviously, $N^{\prime}{\subset}N$.
\\
\\
Moreover, the reader may easily verify that the function $f$: $N\longrightarrow N^{\prime}$ such that for all $l{\in}N$, if $l{\not\in}N^{j}$ then $f(l){=}l$ else $f(l){=}f_{ij}^{{-}1}(l)$ is a $P$-embedding of $(N,{E},\lambda)$ into $(N^{\prime},{E^{\prime}},\lambda^{\prime})$.
\\
\\
In other respect, the reader may easily verify that the function $f^{\prime}$: $N^{\prime}\longrightarrow N$ such that for all $l{\in}N^{\prime}$, $f^{\prime}(l){=}l$ is a $P$-embedding of $(N^{\prime},{E^{\prime}},\lambda^{\prime})$ into $(N,{E},\lambda)$.
\\
\\
Hence, $(N,{E},\lambda){\sim}(N^{\prime},{E^{\prime}},\lambda^{\prime})$.
\medskip
\end{proof}
%
%
%Lemma~\ref{lemma:about:nice:A:trees:and:relation:equivalence:sim:third} asserts that every strict $P$-labelled tree is $\sim$-equivalent to a nice $P$-labelled tree.
%
%
\begin{lemma}\label{lemma:about:nice:A:trees:and:relation:equivalence:sim:third}
For all strict $P$-labelled trees $(N,{E},\lambda)$, there exists a nice $P$-labelled tree $(N^{\prime},{E^{\prime}},\lambda^{\prime})$ such that $(N,{E},\lambda){\sim}(N^{\prime},{E^{\prime}},\lambda^{\prime})$.
\end{lemma}
\begin{proof}
By Lemma~\ref{lemma:transforming:strict:labelled:trees:into:nice:labelled:trees}.
\medskip
\end{proof}
\begin{lemma}\label{lemma:the:quotient:by:simeq:of:the:set:of:all:nice:trees:is:finite}
Modulo $\simeq$, there exists only a finite number of nice $P$-labelled trees.
\end{lemma}
\begin{proof}
Reminding that for all strict $P$-labelled trees $(N,{E},\lambda)$, the height of the tree $(N,E)$ does not exceed $\card(P)$, the reader may easily verify that for all $h{\in}\N$, if $h{\leq}\card(P)$ then modulo $\simeq$, the number $\nPlt(h)$ of nice $P$-labelled trees of height at most $h$ is such that if $h{=}0$ then $\nPlt(h){=}\card(P)$ else $\nPlt(h){=}\card(P){\times}2^{\nPlt(h{-}1)}$.
Hence, modulo $\simeq$, there exists only a finite number of nice $P$-labelled trees.
\medskip
\end{proof}
\begin{lemma}\label{lemma:about:nice:A:trees:quotient:is:finite}
Modulo $\sim$, there exists only a finite number of $P$-labelled trees.
\end{lemma}
\begin{proof}
By Lemmas~\ref{lemma:about:nice:A:trees:and:relation:equivalence:sim}, \ref{lemma:about:nice:A:trees:and:relation:equivalence:sim:third} and~\ref{lemma:the:quotient:by:simeq:of:the:set:of:all:nice:trees:is:finite} and the fact that $\simeq$ is finer than $\sim$.
\medskip
\end{proof}
\begin{lemma}\label{proposition:about:labelled:trees:first}
For all countably indexed families $(N_{m},{E_{m}},\lambda_{m})_{m{\in}\N}$ of $P$-label\-led trees, there exists a strictly increasing function $\theta$: $\N{\longrightarrow}\N$ such that for all $m,n{\in}\N$, $(N_{\theta(m)},{E_{\theta(m)}},\lambda_{\theta(m)}){\sim}(N_{\theta(n)},{E_{\theta(n)}},\lambda_{\theta(n)})$.
\end{lemma}
\begin{proof}
By Lemma~\ref{lemma:about:nice:A:trees:quotient:is:finite}.
\medskip
\end{proof}
Now, we present our promised result about $P$-labelled trees.
\begin{lemma}\label{proposition:about:labelled:trees:to:be:dreary}
For all countably indexed families $(N_{m},{E_{m}},\lambda_{m})_{m{\in}\N}$ of $P$-label\-led trees, there exists $n{\in}\N$ such that the finitely indexed family $(N_{m},{E_{m}},\lambda_{m})_{m{\in}(n)}$ of $P$-labelled trees is dreary.
\end{lemma}
\begin{proof}
By Lemma~\ref{proposition:about:labelled:trees:first}.
\medskip
\end{proof}
\section{Syntax}\label{section:closed:sets:of:formulas}
Let $\At$ be a countably infinite set (with typical members called {\em atoms}\/ and denoted $p$, $q$, etc).
\\
\\
Let $\Fo$ be the countably infinite set (with typical members called {\em formulas}\/ and denoted $A$, $B$, etc) of finite words over $\At{\cup}\{{\rightarrow},{\top},{\bot},{\wedge},{\vee},{\square},{\lozenge},(,)\}$ defined by\footnote{Notice that each connective but $\rightarrow$ has a dual one.
Imitating the ``Brouwerian implication'' introduced by Rauszer~\cite{Rauszer:1980} within the context $\IPL$, a dual implication could be added to the language of $\LIK4$ as well.
We will address this possibility in future work.}
$$A\ {::=}\ p{\mid}(A{\rightarrow}A){\mid}{\top}{\mid}{\bot}{\mid}(A{\wedge}A){\mid}(A{\vee}A){\mid}{\square}A{\mid}{\lozenge}A$$
where $p$ ranges over $\At$.
\\
\\
For all $A{\in}\Fo$, the {\em length of $A$}\/ (denoted ${\parallel}A{\parallel}$) is the number of symbols in $A$.
\\
\\
We follow the standard rules for omission of the parentheses.
\\
\\
For all $A{\in}\Fo$, when we write ``$\neg A$'', we mean ``$A{\rightarrow}{\bot}$''.
\\
\\
Let ${\bowtie}$ be the {\em binary relation of accessibility between sets of formulas}\/ such that for all sets $\Gamma,\Delta$ of formulas, $\Gamma{\bowtie}\Delta$ if and only if for all $A{\in}\Fo$, the following conditions hold:\footnote{Obviously, there exist sets $\Gamma,\Delta$ of formulas such that $\Gamma{\bowtie}\Delta$ and $\Delta{\not\bowtie}\Gamma$.
However, we use a symmetric symbol to denote the binary relation of accessibility between sets of formulas, seeing that in its definition, $\square$-formulas and $\lozenge$-formulas play symmetric roles.}
\begin{itemize}
\item if ${\square}A{\in}\Gamma$ then $A{\in}\Delta$,
\item if $A{\in}\Delta$ then ${\lozenge}A{\in}\Gamma$.
\end{itemize}
A set $\Sigma$ of formulas is {\em closed}\/ if for all $A,B{\in}\Fo$,
\begin{itemize}
\item if $A{\rightarrow}B{\in}\Sigma$ then $A{\in}\Sigma$ and $B{\in}\Sigma$,
\item if $A{\wedge}B{\in}\Sigma$ then $A{\in}\Sigma$ and $B{\in}\Sigma$,
\item if $A{\vee}B{\in}\Sigma$ then $A{\in}\Sigma$ and $B{\in}\Sigma$,
\item if ${\square}A{\in}\Sigma$ then $A{\in}\Sigma$,
\item if ${\lozenge}A{\in}\Sigma$ then $A{\in}\Sigma$.
\end{itemize}
For all $A{\in}\Fo$, let $\Sigma_{A}$ be the least closed set of formulas containing $A$.
\begin{lemma}\label{lemma:cardinal:sigma:A}
For all $A{\in}\Fo$, $\card(\Sigma_{A}){\leq}{\parallel}A{\parallel}$.
\end{lemma}
\begin{proof}
By induction on $A{\in}\Fo$.
\medskip
\end{proof}
For all $A{\in}\Fo$, let $P_{A}{=}\{{-}1\}{\cup}\wp(\Sigma_{A})$.
\begin{lemma}
For all $A{\in}\Fo$, $\card(P_{A}){\leq}1{+}2^{{\parallel}A{\parallel}}$.
\end{lemma}
\begin{proof}
By Lemma~\ref{lemma:cardinal:sigma:A}.
\medskip
\end{proof}
For all $A{\in}\Fo$, let $\leq_{A}$ be the least partial order on $P_{A}$ such that
\begin{itemize}
\item for all $\Delta{\in}\wp(\Sigma_{A})$, ${-}1{\leq_{A}}\Delta$,
\item for all $\Gamma,\Delta{\in}\wp(\Sigma_{A})$, if $\Gamma{\subseteq}\Delta$ then $\Gamma{\leq_{A}}\Delta$.
\end{itemize}
\section{Relational semantics}\label{section:semantics}
A {\em frame}\/ is a relational structure of the form $(W,{\leq},{R})$ where $W$ is a nonempty set, $\leq$ is a preorder on $W$ and ${R}$ is a binary relation on $W$.
\\
\\
A frame $(W,{\leq},{R})$ is {\em finite}\/ if $W$ is finite.
\\
\\
A frame $(W,{\leq},{R})$ is {\em transitive}\/ if $R$ is transitive.
\\
\\
A frame $(W,{\leq},{R})$ is {\em upward confluent}\/ if ${R}{\circ}{\geq}{\subseteq}{\geq}{\circ}{R}$, that is to say for all $s,t,u{\in}W$, if $s{R}t$ and $u{\leq}t$ then there exists $v{\in}W$ such that $v{\leq}s$ and $v{R}u$.
\\
\\
A frame $(W,{\leq},{R})$ is {\em downward confluent}\/ if ${\leq}{\circ}{R}{\subseteq}{R}{\circ}{\leq}$, that is to say for all $s,t,u{\in}W$, if $s{\leq}t$ and $t{R}u$ then there exists $v{\in}W$ such that $s{R}v$ and $v{\leq}u$.
\\
\\
A frame $(W,{\leq},{R})$ is {\em forward confluent}\/ if ${\geq}{\circ}{R}{\subseteq}{R}{\circ}{\geq}$, that is to say for all $s,t,u{\in}W$, if $t{\leq}s$ and $t{R}u$ then there exists $v{\in}W$ such that $s{R}v$ and $u{\leq}v$.
\\
\\
The elementary conditions characterizing upward confluent frames and downward confluent frames have been considered in~\cite{Bozic:Dosen:1984} where they have received no specific name.
The elementary condition characterizing forward confluent frames has been considered in~\cite{FischerServi:1984} where it has been called ``connecting property'' and in~\cite[Chapter~$3$]{Simpson:1994} where it has been called ``$(\mathbf{F}1)$''.
These $3$~elementary conditions have also been considered in~\cite{Plotkin:Stirling:1986}.
\\
\\
A {\em valuation on a frame $(W,{\leq},{R})$}\/ is a function $V$: $\At\longrightarrow\wp(W)$ such that for all $p{\in}\At$, $V(p)$ is $\leq$-closed.
\\
\\
A {\em model}\/ is a $4$-tuple consisting of the $3$ components of a frame and a valuation on that frame.
\\
\\
With respect to a model $(W,{\leq},{R},V)$, for all $s{\in}W$ and for all $A{\in}\Fo$, the {\em satisfiability of $A$ at $s$ in $(W,{\leq},{R},V)$}\/ (in symbols $(W,{\leq},{R},V),s{\models}A$) is inductively defined as follows:
\begin{itemize}
\item $(W,{\leq},{R},V),s{\models}p$ if and only if $s{\in}V(p)$,
\item $(W,{\leq},{R},V),s{\models}A{\rightarrow}B$ if and only if for all $t{\in}W$, if $s{\leq}t$ and $(W,{\leq},{R},V),t
$\linebreak$
{\models}A$ then $(W,{\leq},{R},V),t{\models}B$,
\item $(W,{\leq},{R},V),s{\models}{\top}$,
\item $(W,{\leq},{R},V),s{\not\models}{\bot}$,
\item $(W,{\leq},{R},V),s{\models}A{\wedge}B$ if and only if $(W,{\leq},{R},V),s{\models}A$ and $(W,{\leq},{R},V),s
$\linebreak$
{\models}B$,
\item $(W,{\leq},{R},V),s{\models}A{\vee}B$ if and only if either $(W,{\leq},{R},V),s{\models}A$, or $(W,{\leq},{R},
$\linebreak$
V),s{\models}B$,
\item $(W,{\leq},{R},V),s{\models}{\square}A$ if and only if for all $t{\in}W$, if $s{R}t$ then $(W,{\leq},{R},V),t{\models}
$\linebreak$
A$,
\item $(W,{\leq},{R},V),s{\models}{\lozenge}A$ if and only if there exists $t{\in}W$ such that $s{R}t$ and $(W,{\leq},{R},V),t{\models}A$.
\end{itemize}
For all models $(W,{\leq},{R},V)$, for all $s{\in}W$ and for all $A{\in}\Fo$, if $(W,{\leq},{R},V)$ is clear from the context then when we write ``$s{\models}A$'', we mean ``$(W,{\leq},{R},V),s{\models}A$''.
\begin{lemma}[Heredity Property]\label{lemma:heredity:property}
Let $(W,{\leq},{R},V)$ be a model based on a
\linebreak
downward confluent and forward confluent frame.
For all $A{\in}\Fo$ and for all $s,t{\in}W$, if $s{\models}A$ and $s{\leq}t$ then $t{\models}A$.\footnote{The heredity property is probably the most fundamental property of the relational semantics of intuitionistically-based logics.}
\end{lemma}
\begin{proof}
By induction on $A{\in}\Fo$.
\medskip
\end{proof}
Our definition of the satisfiability of formulas is the one introduced by Bo\v{z}i\'{c} and Do\v{s}en~\cite{Bozic:Dosen:1984}.
\\
\\
Fischer Servi~\cite{FischerServi:1984} defines the satisfiability of $\square$-formulas and $\lozenge$-formulas as follows: $s{\models_{\mathtt{FS}}}{\square}A$ if and only if for all $t{\in}W$, if $s{\leq}t$ then for all $u{\in}W$, if $t{R}u$ then $u{\models_{\mathtt{FS}}}A$; $s{\models_{\mathtt{FS}}}{\lozenge}A$ if and only if there exists $t{\in}W$ such that $s{R}t$ and $t{\models_{\mathtt{FS}}}A$.
\\
\\
P\v{r}enosil~\cite{Prenosil:2014} defines the satisfiability of $\square$-formulas and $\lozenge$-formulas as follows: $s{\models_{\mathtt{P}}}{\square}A$ if and only if for all $t{\in}W$, if $s{\leq}t$ then for all $u{\in}W$, if $t{R}u$ then $u{\models_{\mathtt{P}}}A$; $s{\models_{\mathtt{P}}}{\lozenge}A$ if and only if there exists $t{\in}W$ such that $s{\geq}t$ and there exists $u{\in}W$ such that $t{R}u$ and $u{\models_{\mathtt{P}}}A$.
\\
\\
Wijesekera~\cite{Wijesekera:1990} defines the satisfiability of $\square$-formulas and $\lozenge$-formulas as follows: $s{\models_{\mathtt{W}}}{\square}A$ if and only if for all $t{\in}W$, if $s{\leq}t$ then for all $u{\in}W$, if $t{R}u$ then $u{\models_{\mathtt{W}}}A$; $s{\models_{\mathtt{W}}}{\lozenge}A$ if and only if for all $t{\in}W$, if $s{\leq}t$ then there exists $u{\in}W$ such that $t{R}u$ and $u{\models_{\mathtt{W}}}A$.
\\
\\
In the class of all downward confluent and forward confluent frames, the definition of the satisfiability of formulas considered by Fischer Servi, the definition of the satisfiability of formulas considered by P\v{r}enosil, the definition of the satisfiability of formulas considered by Wijesekera and the definition of the satisfiability of formulas considered by Bo\v{z}i\'{c} and Do\v{s}en give rise to the same relation of satisfiability.
\\
\\
A formula $A$ is {\em true in a model $(W,{\leq},{R},V)$}\/ (in symbols $(W,{\leq},{R},V){\models}A$) if for all $s{\in}W$, $s{\models}A$.
\\
\\
A formula $A$ is {\em valid in a frame $(W,{\leq},{R})$}\/ (in symbols $(W,{\leq},{R}){\models}A$) if for all models $(W,{\leq},{R},V)$ based on $(W,{\leq},{R})$, $(W,{\leq},{R},V){\models}A$.
\section{Axiomatization}\label{section:axiomatization}
See~\cite[Chapter~$2$]{Chagrov:Zakharyaschev:1997} for an introduction to the standard axioms of $\IPL$ and the standard inference rules of $\IPL$.
\\
\\
Let $\LIK4$ be the least set of formulas closed under uniform substitution, containing the standard axioms of $\IPL$, closed under the standard inference rules of $\IPL$, containing the following axioms and closed under the following inference rules:\footnote{The reason why we have separated the condition of closure under uniform substitution and the condition of closure under the standard inference rules of $\IPL$ is simply the following: uniform substitution is not an inference rule.}
\begin{description}
\item[$(\mathbf{C}\square)$] ${\square}p{\wedge}{\square}q{\rightarrow}{\square}(p{\wedge}q)$,
\item[$(\mathbf{C}\lozenge)$] ${\lozenge}(p{\vee}q){\rightarrow}{\lozenge}p{\vee}{\lozenge}q$,
\item[$(\axiomN{\square})$] ${\square}{\top}$,
\item[$(\axiomN{\lozenge})$] ${\neg}{\lozenge}{\bot}$,
\item[$(\Rule{\square})$] $\frac{p{\rightarrow}q}{{\square}p{\rightarrow}{\square}q}$,
\item[$(\Rule{\lozenge})$] $\frac{p{\rightarrow}q}{{\lozenge}p{\rightarrow}{\lozenge}q}$,
\item[$(4\square)$] ${\square}p{\rightarrow}{\square}{\square}p$,
\item[$(4\lozenge)$] ${\lozenge}{\lozenge}p{\rightarrow}{\lozenge}p$,
\item[$(\Axiom\mathbf{d})$] ${\square}(p{\vee}q){\rightarrow}{\lozenge}p{\vee}{\square}q$,
\item[$(\Axiom\mathbf{f})$] ${\lozenge}(p{\rightarrow}q){\rightarrow}({\square}p{\rightarrow}{\lozenge}q)$.
\end{description}
The axioms $(\mathbf{C}\square)$, $(\mathbf{C}\lozenge)$, $(\axiomN{\square})$ and $(\axiomN{\lozenge})$ and the inference rules $(\Rule{\square})$ and $(\Rule{\lozenge})$ are well-known.
They are sometimes associated to the concept of normality in classical modal logics.
The axioms $(4\square)$ and $(4\lozenge)$ are also well-known.
They are classically related to the elementary condition of transitivity.
The axiom $(\Axiom\mathbf{d})$ has been already considered by Bo\v{z}i\'{c} and Do\v{s}en~\cite{Bozic:Dosen:1984}.
It is related to the elementary condition of downward confluence.
The axiom $(\Axiom\mathbf{f})$ has been already considered by Fischer Servi~\cite{FischerServi:1984}.
It is related to the elementary condition of forward confluence.
\\
\\
The reader may easily see that the axiom ${\square}(p{\vee}q){\rightarrow}(({\lozenge}p{\rightarrow}{\square}q){\rightarrow}{\square}q)$ and the inference rule $\frac{{\lozenge}p{\rightarrow}q{\vee}{\square}(p{\rightarrow}r)}{{\lozenge}p{\rightarrow}q{\vee}{\lozenge}r}$ considered in Lemmas~\ref{lemma:about:wCD} and~\ref{lemma:about:the:special:rule:of:inference} have similarities with the equation ${\lozenge}a{\rightarrow}{\square}b{\leq}{\square}(a{\vee}b){\rightarrow}{\square}b$ and the law ${\lozenge}b{\leq}{\square}a{\vee}c{\Rightarrow}{\lozenge}b
$\linebreak$
{\leq}{\lozenge}(a{\wedge}b){\vee}c$ considered by P\v{r}enosil~\cite{Prenosil:2014}.
Together, they will allow us to use the canonical model construction developed in~\cite{Balbiani:Gencer:preliminary:draft:SL}.
\begin{lemma}\label{lemma:about:wCD}
The axiom ${\square}(p{\vee}q){\rightarrow}(({\lozenge}p{\rightarrow}{\square}q){\rightarrow}{\square}q)$ is derivable in $\LIK4$.
\end{lemma}
\begin{proof}
By considering for all $A,B{\in}\Fo$, the following sequence of formulas, the reader may easily construct a proof of the derivability of the axiom ${\square}(p{\vee}q){\rightarrow}
$\linebreak$
(({\lozenge}p{\rightarrow}{\square}q){\rightarrow}{\square}q)$ in $\LIK4$:\footnote{We invite the reader to notice where axiom $(\Axiom\mathbf{d})$ is used in this sequence.}
\begin{enumerate}
\item ${\lozenge}A{\vee}{\square}B{\rightarrow}(({\lozenge}A{\rightarrow}{\square}B){\rightarrow}{\square}B)$ ($\IPL$-reasoning),
\item ${\square}(A{\vee}B){\rightarrow}{\lozenge}A{\vee}{\square}B$ (axiom $(\Axiom\mathbf{d})$),
\item ${\square}(A{\vee}B){\rightarrow}(({\lozenge}A{\rightarrow}{\square}B){\rightarrow}{\square}B)$ ($\IPL$-reasoning on~$\mathbf{(1)}$ and~$\mathbf{(2)}$).
\end{enumerate}
\medskip
\end{proof}
\begin{lemma}\label{lemma:about:the:special:rule:of:inference}
The inference rule $\frac{{\lozenge}p{\rightarrow}q{\vee}{\square}(p{\rightarrow}r)}{{\lozenge}p{\rightarrow}q{\vee}{\lozenge}r}$ is derivable in $\LIK4$.
\end{lemma}
\begin{proof}
By considering for all $A,B,C{\in}\Fo$, the following sequence of formulas, the reader may easily construct a proof of the derivability of the inference rule $\frac{{\lozenge}p{\rightarrow}q{\vee}{\square}(p{\rightarrow}r)}{{\lozenge}p{\rightarrow}q{\vee}{\lozenge}r}$ in $\LIK4$:\footnote{We invite the reader to notice where axiom $(\Axiom\mathbf{f})$ is used in this sequence.}
\begin{enumerate}
\item ${\lozenge}A{\rightarrow}B{\vee}{\square}(A{\rightarrow}C)$ (hypothesis),
\item $A{\rightarrow}((A{\rightarrow}C){\rightarrow}C)$ ($\IPL$-reasoning),
\item ${\lozenge}A{\rightarrow}{\lozenge}((A{\rightarrow}C){\rightarrow}C)$ (inference rule $(\Rule{\lozenge})$ on~$\mathbf{(2)}$),
\item ${\lozenge}((A{\rightarrow}C){\rightarrow}C){\rightarrow}({\square}(A{\rightarrow}C){\rightarrow}{\lozenge}C)$ (axiom $(\Axiom\mathbf{f})$),
\item ${\lozenge}A{\rightarrow}B{\vee}{\lozenge}C$ ($\IPL$-reasoning on~$\mathbf{(1)}$, $\mathbf{(3)}$ and~$\mathbf{(4)}$).
\end{enumerate}
\medskip
\end{proof}
\begin{lemma}[Soundness]\label{proposition:soundness:of:the:logics}
Let $A$ be a formula.
If $A{\in}\LIK4$ then for all transitive, downward confluent and forward confluent frames $(W,{\leq},{R})$, $(W,{\leq},{R}){\models}A$.
\end{lemma}
\section{Canonical model}\label{section:canonical:model}
A {\em theory}\/ is a set of formulas containing $\LIK4$ and closed with respect to the inference rule of modus ponens.
\\
\\
A theory $s$ is {\em proper}\/ if ${\bot}{\not\in}s$.
\\
\\
A proper theory $s$ is {\em prime}\/ if for all $A,B{\in}\Fo$, if $A{\vee}B{\in}s$ then either $A{\in}s$, or $B{\in}s$.
\begin{lemma}[Lindenbaum Lemma]\label{lemma:almost:completeness}
Let $A{\in}\Fo$.
If $A{\not\in}\LIK4$ then there exists a prime theory $s$ such that $A{\not\in}s$.
\end{lemma}
\begin{proof}
See~\cite[Lemma~$24$]{Balbiani:Gencer:preliminary:draft:SL}.
\medskip
\end{proof}
Let $(W_{c},{\leq_{c}},{R_{c}})$ be the frame such that
\begin{itemize}
\item $W_{c}$ is the nonempty set of all prime theories,\footnote{By Lemma~\ref{proposition:soundness:of:the:logics}, ${\bot}{\not\in}\LIK4$.
As a result, by Lemma~\ref{lemma:almost:completeness}, there exists a prime theory.
Therefore, $W_{c}$ is a nonempty set.}
\item $\leq_{c}$ is the preorder on $W_{c}$ such that for all $s,t{\in}W_{c}$, $s{\leq_{c}}t$ if and only if $s{\subseteq}t$,\footnote{Obviously, $\leq_{c}$ is a reflexive and transitive binary relation on $W_{c}$.
As a result, $\leq_{c}$ is a preorder on $W_{c}$.}
\item $R_{c}$ is the binary relation on $W_{c}$ such that for all $s,t{\in}W_{c}$, $s{R_{c}}t$ if and only if $s{\bowtie}t$.
\end{itemize}
The frame $(W_{c},{\leq_{c}},{R_{c}})$ is called {\em canonical frame of $\LIK4$.}
\begin{lemma}\label{lemma:canonical:frame:confluences}
$(W_{c},{\leq_{c}},{R_{c}})$ is transitive, downward confluent and forward confluent.
\end{lemma}
\begin{proof}
About downward confluence, it follows from Lemmas~\ref{lemma:about:wCD} and~\ref{lemma:about:the:special:rule:of:inference} and~\cite[Lemma~$29$]{Balbiani:Gencer:preliminary:draft:SL}.
As for forward confluence, it follows from Lemmas~\ref{lemma:about:wCD} and~\ref{lemma:about:the:special:rule:of:inference} and~\cite[Lemma~$27$]{Balbiani:Gencer:preliminary:draft:SL}.
\medskip
\end{proof}
The {\em canonical valuation of $\LIK4$}\/ is the valuation $V_{c}$: $\At\longrightarrow\wp(W_{c})$ on $(W_{c},{\leq_{c}},
$\linebreak$
{R_{c}})$ such that for all $p{\in}\At$, $V_{c}(p){=}\{s{\in}W_{c}$: $p{\in}s\}$.\footnote{Obviously, for all $p{\in}\At$, $V_{c}(p)$ is $\leq_{c}$-closed.}
\\
\\
The {\em canonical model of $\LIK4$}\/ is the model $(W_{c},{\leq_{c}},{R_{c}},V_{c})$.
\begin{lemma}[Canonical Truth Lemma]\label{lemma:truth:lemma}
For all $A{\in}\Fo$ and for all $s{\in}W_{c}$, $A{\in}s$ if and only if $(W_{c},{\leq_{c}},{R_{c}},V_{c}),s{\models}A$.
\end{lemma}
\begin{proof}
See~\cite[Lemma~$26$]{Balbiani:Gencer:preliminary:draft:SL}.
\medskip
\end{proof}
\section{Maximality}\label{section:maximal:worlds}
We say that $s{\in}W_{c}$ is {\em maximal with respect to $B{\in}\Fo$}\/ if for all $t{\in}W_{c}$, if $s{<_{c}}t$ then $t{\models}B$.
\begin{lemma}\label{lemma:maximal}
Let $s{\in}W_{c}$ and $B{\in}\Fo$.
If $s$ is not maximal with respect to $B$ then there exists $t{\in}W_{c}$ such that $s{<_{c}}t$, $t{\not\models}B$ and $t$ is maximal with respect to $B$.
\end{lemma}
\begin{proof}
Suppose $s$ is not maximal with respect to $B$.
Let $S=\{t{\in}W_{c}$: $s{<_{c}}t$ and $t{\not\models}B\}$.
Since $s$ is not maximal with respect to $B$, then $S$ is nonempty.
Moreover, obviously, for all nonempty chains $(t_{a})_{a{\in}I}$ of elements of $S$, $\bigcup\{t_{a}$: $a{\in}I\}$ is in $S$.\footnote{Remind that for all $a{\in}I$, $s{<_{c}}t_{a}$ and $t_{a}{\not\models}B$, i.e. $s{\subseteq}t_{a}$, $t_{a}{\not\subseteq}s$ and $B{\not\in}t_{a}$.
Moreover, for all $b,c{\in}I$, either $t_{b}{\subseteq}t_{c}$, or $t_{c}{\subseteq}t_{b}$.
As a result, $s{\subseteq}\bigcup\{t_{a}$: $a{\in}I\}$, $\bigcup\{t_{a}$: $a{\in}I\}{\not\subseteq}s$ and $B{\not\in}\bigcup\{t_{a}$: $a{\in}I\}$, i.e. $s{<_{c}}\bigcup\{t_{a}$: $a{\in}I\}$ and $\bigcup\{t_{a}$: $a{\in}I\}{\not\models}B$.}
Hence, by Zorn's Lemma, $S$ possesses a maximal element $t$.\footnote{See~\cite[Chapter~$10$]{Davey:Priestley:2002} and~\cite[Chapter~$1$]{Wechler:1992} for details about Zorn's Lemma.}
Obviously, $s{<_{c}}t$, $t{\not\models}B$ and $t$ is maximal with respect to $B$.
\medskip
\end{proof}
\begin{lemma}\label{lemma:maximal:rightarrow}
Let $s{\in}W_{c}$ and $B,C{\in}\Fo$.
If $s{\not\models}B{\rightarrow}C$ and $s$ is maximal with respect to $B{\rightarrow}C$ then $s{\models}B$ and $s{\not\models}C$.
\end{lemma}
\begin{proof}
Suppose $s{\not\models}B{\rightarrow}C$ and $s$ is maximal with respect to $B{\rightarrow}C$.
For the sake of the contradiction, suppose either $s{\not\models}B$, or $s{\models}C$.
Since $s{\not\models}B{\rightarrow}C$, then there exists $t{\in}W_{c}$ such that $s{\leq_{c}}t$, $t{\models}B$ and $t{\not\models}C$.
Hence, either $s{=}t$, or $s{<_{c}}t$.
In the former case, since $t{\models}B$ and $t{\not\models}C$, then $s{\models}B$ and $s{\not\models}C$.
Since either $s{\not\models}B$, or $s{\models}C$, then $s{\models}C$: a contradiction.
In the latter case, since $s$ is maximal with respect to $B{\rightarrow}C$, then $t{\models}B{\rightarrow}C$.
Since $t{\models}B$, then $t{\models}C$: a contradiction.
\medskip
\end{proof}
\section{Tips and clips}\label{section:tips}
From now on in this note and up to Lemma~\ref{lemma:truth:lemma:at:the:end:of:the:procedure}, let $A{\in}\Fo$ and $s_{0}$ be a prime theory.
\\
\\
A {\em tip}\/ is a $4$-tuple of the form $(i,s,\alpha,X)$ where $i{\in}\N$, $s{\in}W_{c}$, $\alpha{\in}\N$ and $X{\in}\N$.
\\
\\
The tip $(0,s_{0},0,0)$ is called {\em initial tip of $s_{0}$.}
\\
\\
The {\em rank of the tip $(i,s,\alpha,X)$}\/ is $\alpha$.
\\
\\
The {\em height of the tip $(i,s,\alpha,X)$}\/ is $X$.
\\
\\
The {\em degree of the tip $(i,s,\alpha,X)$}\/ (denoted $\degre(i,s,\alpha,X)$) is the number of $B{\in}\Sigma_{A}$ such that $s{\not\models}B$ and $s$ is not maximal with respect to $B$.
\begin{lemma}\label{lemma:about:degree:and:Sigma:A}
For all tips $(i,s,\alpha,X)$, $\degre(i,s,\alpha,X){\leq}\card(\Sigma_{A})$.
\end{lemma}
For all finite and nonempty sets ${\mathcal T}$ of tips and for all $\alpha{\in}\N$, let $\hauteur_{\alpha}({\mathcal T}){=}\max\{X{\in}\N$: $(i,s,\alpha,X){\in}{\mathcal T}\}$.
\\
\\
A {\em clip}\/ is a relational structure of the form $({\mathcal T},{\ll},{\triangleright})$ where ${\mathcal T}$ is a finite and nonempty set of tips and $\ll$ and $\triangleright$ are binary relations on $\mathcal T$.
\\
\\
The clip $(\{(0,s_{0},0,0)\},\emptyset,\emptyset)$ is called {\em initial clip of $s_{0}$.}
\\
\\
We say that the clip $({\mathcal T},{\ll},{\triangleright})$ is {\em coherent}\/ if for all $(i,s,\alpha,X),(j,t,\beta,Y){\in}{\mathcal T}$,
\begin{itemize}
\item if $i{=}j$ then $s{=}t$, $\alpha{=}\beta$ and $X{=}Y$,
\item if $(i,s,\alpha,X){\ll}(j,t,\beta,Y)$ then $j{\not=}i$, $s{\leq_{c}}t$, $\beta{=}\alpha$ and $Y{=}X{+}1$,
\item if $(i,s,\alpha,X){\triangleright}(j,t,\beta,Y)$ then $j{\not=}i$, $t{\in}R_{c}(s)$, $\beta{=}\alpha{+}1$ and $Y{=}X$.
\end{itemize}
\begin{lemma}\label{lemma:about:initial:clip:coherent}
The initial clip of $s_{0}$ is coherent.
\end{lemma}
Obviously, for all clips $({\mathcal T},{\ll},{\triangleright})$, the frame $({\mathcal T},{\ll^{\star}},{\triangleright^{+}})$ is finite and transitive.
\begin{lemma}\label{lemma:about:coherence:and:morphisms}
Let $({\mathcal T},{\ll},{\triangleright})$ be a clip.
If $({\mathcal T},{\ll},{\triangleright})$ is coherent then the function $f$: ${\mathcal T}{\longrightarrow}W_{c}$ such that for all $(i,s,\alpha,X){\in}{\mathcal T}$, $f(i,s,\alpha,X){=}s$ is a homomorphism from the frame $({\mathcal T},{\ll^{\star}},{\triangleright^{+}})$ to the canonical frame of $\LIK4$.
\end{lemma}
\begin{proof}
Suppose $({\mathcal T},{\ll},{\triangleright})$ is coherent.
Let $(i,s,\alpha,X),(j,t,\beta,Y){\in}{\mathcal T}$.
\\
\\
Firstly, suppose $(i,s,\alpha,X){\ll^{\star}}(j,t,\beta,Y)$.
Hence, there exists $n{\in}\N$ and there exists $(k_{0},u_{0},\gamma_{0},Z_{0}),\ldots,(k_{n},u_{n},\gamma_{n},Z_{n}){\in}{\mathcal T}$ such that $(k_{0},u_{0},\gamma_{0},Z_{0}){=}(i,s,\alpha,X)$, $(k_{n},u_{n},\gamma_{n},Z_{n}){=}(j,t,\beta,Y)$ and for all $a{\in}(n)$, $(k_{a{-}1},u_{a{-}1},\gamma_{a{-}1},Z_{a{-}1}){\ll}(k_{a},u_{a},
$\linebreak$
\gamma_{a},Z_{a})$.
By induction on $n{\in}\N$, the reader may easily verify that $s{\leq_{c}}t$.
\\
\\
Secondly, suppose $(i,s,\alpha,X){\triangleright^{+}}(j,t,\beta,Y)$.
Thus, there exists $n{\geq}1$ and there exists $(k_{0},u_{0},\gamma_{0},Z_{0}),\ldots,(k_{n},u_{n},\gamma_{n},Z_{n}){\in}{\mathcal T}$ such that $(k_{0},u_{0},\gamma_{0},Z_{0}){=}(i,s,\alpha,X)$, $(k_{n},u_{n},\gamma_{n},Z_{n}){=}(j,t,\beta,Y)$ and for all $a{\in}(n)$, $(k_{a{-}1},u_{a{-}1},\gamma_{a{-}1},Z_{a{-}1}){\triangleright}(k_{a},u_{a},
$\linebreak$
\gamma_{a},Z_{a})$.
By induction on $n{\geq}1$, the reader may easily verify that $s{R_{c}}t$.
\medskip
\end{proof}
We say that the coherent clip $({\mathcal T},{\ll},{\triangleright})$ is {\em regular}\/ if for all $(i,s,\alpha,X),(j,t,\beta,Y),
$\linebreak$
(k,u,\gamma,Z){\in}{\mathcal T}$,
\begin{itemize}
\item if $(i,s,\alpha,X){\ll}(k,u,\gamma,Z)$ and $(j,t,\beta,Y){\ll}(k,u,\gamma,Z)$ then $i{=}j$,
\item if $(i,s,\alpha,X){\triangleright}(k,u,\gamma,Z)$ and $(j,t,\beta,Y){\triangleright}(k,u,\gamma,Z)$ then $i{=}j$,
\item if $(i,s,\alpha,X){\triangleright}(j,t,\beta,Y)$ and $(k,u,\gamma,Z){\ll}(j,t,\beta,Y)$ then there exists $(l,v,
$\linebreak$
\delta,T){\in}{\mathcal T}$ such that $(l,v,\delta,T){\ll}(i,s,\alpha,X)$ and $(l,v,\delta,T){\triangleright}(k,u,\gamma,Z)$.
\end{itemize}
\begin{lemma}\label{lemma:about:initial:clip:regular}
The initial clip of $s_{0}$ is regular.
\end{lemma}
Obviously, for all regular clips $({\mathcal T},{\ll},{\triangleright})$, the relational structure $({\mathcal T},{\ll})$ is a disjoint union of trees.
\\
\\
Moreover, for all regular clips $({\mathcal T},{\ll},{\triangleright})$, the frame $({\mathcal T},{\ll^{\star}},{\triangleright^{+}})$ is upward confluent.
\\
\\
For all $\alpha{\in}\N$, the {\em $\alpha$-slice of a regular clip $({\mathcal T},{\ll},{\triangleright})$}\/ is the $P_{A}$-labelled tree $({\mathcal T}_{\alpha},{\ll_{\alpha}},\lambda_{\alpha})$ where\footnote{The reader may easily verify that $({\mathcal T}_{\alpha},{\ll_{\alpha}})$ is a tree and $\lambda_{\alpha}$ is $\leq_{A}$-monotone.}
\begin{itemize}
\item ${\mathcal T}_{\alpha}{=}\{{-}1\}{\cup}\{(j,t,\beta,Y){\in}{\mathcal T}$: $\alpha{=}\beta\}$,
\item $\ll_{\alpha}$ is the least binary relation on ${\mathcal T}_{\alpha}$ such that
\begin{itemize}
\item for all $(k,u,\alpha,Z){\in}{\mathcal T}_{\alpha}$, if for all $(j,t,\alpha,Y){\in}{\mathcal T}_{\alpha}$, $(j,t,\alpha,Y){\not\ll}(k,u,\alpha,Z)$ then ${-}1{\ll_{\alpha}}(k,u,\alpha,Z)$,
\item for all $(j,t,\alpha,Y),(k,u,\alpha,Z){\in}{\mathcal T}_{\alpha}$, if $(j,t,\alpha,Y){\ll}(k,u,\alpha,Z)$ then $(j,t,
$\linebreak$
\alpha,Y){\ll_{\alpha}}(k,u,\alpha,Z)$,
\end{itemize}
\item $\lambda_{\alpha}$: ${\mathcal T}_{\alpha}\longrightarrow P_{A}$ is the function such that
\begin{itemize}
\item $\lambda_{\alpha}({-}1){=}{-}1$,
\item for all $(j,t,\alpha,Y){\in}{\mathcal T}_{\alpha}$, $\lambda_{\alpha}(j,t,\alpha,Y){=}t{\cap}\Sigma_{A}$.
\end{itemize}
\end{itemize}
The {\em finitely indexed family of $P_{A}$-labelled trees associated to a regular clip $({\mathcal T},{\ll},
$\linebreak$
{\triangleright})$ and $\alpha{\in}\N$}\/ is the finitely indexed family $({\mathcal T}_{\beta},{\ll_{\beta}},\lambda_{\beta})_{\beta{\in}(\alpha)}$.
\section{Defects}\label{section:defects}
Let $({\mathcal T},{\ll},{\triangleright})$ be a regular clip.
\\
\\
It might be the case that for some $(i,s,\alpha,X){\in}{\mathcal T}$ and for some $B{\rightarrow}C{\in}\Sigma_{A}$, $s{\not\models}B{\rightarrow}C$ and $s$ is not maximal with respect to $B{\rightarrow}C$ whereas there is no witness in ${\mathcal T}$ of this fact, i.e. there is no $(j,t,\beta,Y){\in}{\mathcal T}$ such that $(i,s,\alpha,X){\ll}(j,t,\beta,Y)$, $t{\models}B$ and $t{\not\models}C$.
Such situation will correspond to a defect of maximality.
\\
\\
Moreover, it might also be the case that for some $(i,s,\alpha,X){\in}{\mathcal T}$ and for some ${\square}B{\in}\Sigma_{A}$, $s{\not\models}{\square}B$ whereas there is no witness in ${\mathcal T}$ of this fact, i.e. there is no $(j,t,\beta,Y){\in}{\mathcal T}$ such that $(i,s,\alpha,X){\triangleright}(j,t,\beta,Y)$ and $t{\not\models}B$.
Such situation will correspond to a defect of $\square$-accessibility.
\\
\\
In other respect, it might also be the case that for some $(i,s,\alpha,X){\in}{\mathcal T}$ and for some ${\lozenge}B{\in}\Sigma_{A}$, $s{\models}{\lozenge}B$ whereas there is no witness in ${\mathcal T}$ of this fact, i.e. there is no $(j,t,\beta,Y){\in}{\mathcal T}$ such that $(i,s,\alpha,X){\triangleright}(j,t,\beta,Y)$ and $t{\models}B$.
Such situation will correspond to a defect of $\lozenge$-accessibility.
\\
\\
Finally, it might as well be the case that either the frame $({\mathcal T},{\ll^{\star}},{\triangleright^{+}})$ is not downward confluent, or the frame $({\mathcal T},{\ll^{\star}},{\triangleright^{+}})$ is not forward confluent.
Such situation will correspond to defects of downward confluence and defects of forward confluence.
\\
\\
A {\em defect of maximality of a regular clip $({\mathcal T},{\ll},{\triangleright})$}\/ is a triple $((i,s,\alpha,X),B,C)$ where $(i,s,\alpha,X){\in}{\mathcal T}$ and $B,C{\in}\Fo$ are such that\footnote{Here, when we write ``$s{\not\models}B{\rightarrow}C$'' and ``$t{\models}B{\rightarrow}C$'', we mean ``$(W_{c},{\leq_{c}},{R_{c}},V_{c}),s{\not\models}B{\rightarrow}C$'' and ``$(W_{c},{\leq_{c}},{R_{c}},V_{c}),t{\models}B{\rightarrow}C$''.}
\begin{itemize}
\item $B{\rightarrow}C{\in}\Sigma_{A}$,
\item $s{\not\models}B{\rightarrow}C$ and $s$ is not maximal with respect to $B{\rightarrow}C$,
\item for all $(j,t,\beta,Y){\in}{\mathcal T}$, if $(i,s,\alpha,X){\ll}(j,t,\beta,Y)$ then $t{\models}B{\rightarrow}C$.
\end{itemize}
The {\em rank of the defect $((i,s,\alpha,X),B,C)$ of maximality}\/ is $\alpha$.
\\
\\
The {\em height of the defect $((i,s,\alpha,X),B,C)$ of maximality}\/ is $X$.
\\
\\
A {\em defect of $\square$-accessibility of a regular clip $({\mathcal T},{\ll},{\triangleright})$}\/ is a couple $((i,s,\alpha,X),B)$ where $(i,s,\alpha,X){\in}{\mathcal T}$ and $B{\in}\Fo$ are such that\footnote{Here, when we write ``$s{\not\models}{\square}B$'' and ``$t{\models}B$'', we mean ``$(W_{c},{\leq_{c}},{R_{c}},V_{c}),s{\not\models}{\square}B$'' and ``$(W_{c},{\leq_{c}},{R_{c}},V_{c}),t{\models}B$''.}
\begin{itemize}
\item ${\square}B{\in}\Sigma_{A}$,
\item $s{\not\models}{\square}B$,
\item for all $(j,t,\beta,Y){\in}{\mathcal T}$, if $(i,s,\alpha,X){\triangleright}(j,t,\beta,Y)$ then $t{\models}B$.
\end{itemize}
The {\em rank of the defect $((i,s,\alpha,X),B)$ of $\square$-accessibility}\/ is $\alpha$.
\\
\\
The {\em height of the defect $((i,s,\alpha,X),B)$ of $\square$-accessibility}\/ is $X$.
\\
\\
A {\em defect of $\lozenge$-accessibility of a regular clip $({\mathcal T},{\ll},{\triangleright})$}\/ is a couple $((i,s,\alpha,X),B)$ where $(i,s,\alpha,X){\in}{\mathcal T}$ and $B{\in}\Fo$ are such that\footnote{Here, when we write ``$s{\models}{\lozenge}B$'' and ``$t{\not\models}B$'', we mean ``$(W_{c},{\leq_{c}},{R_{c}},V_{c}),s{\models}{\lozenge}B$'' and ``$(W_{c},{\leq_{c}},{R_{c}},V_{c}),t{\not\models}B$''.}
\begin{itemize}
\item ${\lozenge}B{\in}\Sigma_{A}$,
\item $s{\models}{\lozenge}B$,
\item for all $(j,t,\beta,Y){\in}{\mathcal T}$, if $(i,s,\alpha,X){\triangleright}(j,t,\beta,Y)$ then $t{\not\models}B$.
\end{itemize}
The {\em rank of the defect $((i,s,\alpha,X),B)$ of $\lozenge$-accessibility}\/ is $\alpha$.
\\
\\
The {\em height of the defect $((i,s,\alpha,X),B)$ of $\lozenge$-accessibility}\/ is $X$.
\\
\\
A {\em defect of downward confluence of a regular clip $({\mathcal T},{\ll},{\triangleright})$}\/ is a triple $((i,s,\alpha,X),
$\linebreak$
(j,t,\beta,Y),(k,u,\gamma,Z))$ where $(i,s,\alpha,X),(j,t,\beta,Y),(k,u,\gamma,Z){\in}{\mathcal T}$ are such that
\begin{itemize}
\item $(i,s,\alpha,X){\ll}(j,t,\beta,Y)$,
\item $(j,t,\beta,Y){\triangleright}(k,u,\gamma,Z)$,
\item for all $(l,v,\delta,T){\in}{\mathcal T}$, either $(i,s,\alpha,X){\not\triangleright}(l,v,\delta,T)$, or $(l,v,\delta,T){\not\ll}(k,u,\gamma,
$\linebreak$
Z)$.
\end{itemize}
The {\em rank of the defect $((i,s,\alpha,X),(j,t,\beta,Y),(k,u,\gamma,Z))$ of downward confluence}\/ is $\alpha$.
\\
\\
The {\em height of the defect $((i,s,\alpha,X),(j,t,\beta,Y),(k,u,\gamma,Z))$ of downward confluence}\/ is $X$.
\begin{lemma}\label{lemma:about:defects:of:downward:confluence}
Let $({\mathcal T},{\ll},{\triangleright})$ be a regular clip.
If $((i,s,\alpha,X),(j,t,\beta,Y),(k,u,\gamma,
$\linebreak$
Z))$ is a defect of downward confluence of $({\mathcal T},{\ll},{\triangleright})$ then for all $(l,v,\delta,T){\in}{\mathcal T}$, $(l,v,\delta,T){\not\ll}(k,u,\gamma,Z)$.
\end{lemma}
\begin{proof}
Suppose $((i,s,\alpha,X),(j,t,\beta,Y),(k,u,\gamma,Z))$ is a defect of downward confluence of $({\mathcal T},{\ll},{\triangleright})$.
For the sake of the contradiction, suppose there exists $(l,v,\delta,T){\in}{\mathcal T}$ such that $(l,v,\delta,T){\ll}(k,u,\gamma,Z)$.
Since $((i,s,\alpha,X),(j,t,\beta,Y),(k,
$\linebreak$
u,\gamma,Z))$ is a defect of downward confluence of $({\mathcal T},{\ll},{\triangleright})$, then $(i,s,\alpha,X){\ll}(j,t,
$\linebreak$
\beta,Y)$ and $(j,t,\beta,Y){\triangleright}(k,u,\gamma,Z)$.
Since $({\mathcal T},{\ll},{\triangleright})$ is regular and $(l,v,\delta,T){\ll}(k,u,
$\linebreak$
\gamma,Z)$, then there exists $(m,w,\epsilon,U){\in}{\mathcal T}$ such that $(m,w,\epsilon,U){\ll}(j,t,\beta,Y)$ and $(m,w,\epsilon,U){\triangleright}(l,v,\delta,T)$.
Since $({\mathcal T},{\ll},{\triangleright})$ is regular and $(i,s,\alpha,X){\ll}(j,t,\beta,Y)$, then $m{=}i$.
Hence, $w{=}s$, $\epsilon{=}\alpha$ and $U{=}X$.
Since $(m,w,\epsilon,U){\triangleright}(l,v,\delta,T)$ and $m{=}i$, then $(i,s,\alpha,X){\triangleright}(l,v,\delta,T)$.
Since $((i,s,\alpha,X),(j,t,\beta,Y),(k,u,\gamma,Z))$ is a defect of downward confluence of $({\mathcal T},{\ll},{\triangleright})$, then either $(i,s,\alpha,X){\not\triangleright}(l,v,\delta,T)$, or $(l,v,\delta,T){\not\ll}(k,u,\gamma,Z)$.
Since $(i,s,\alpha,X){\triangleright}(l,v,\delta,T)$, then $(l,v,\delta,T){\not\ll}(k,u,\gamma,
$\linebreak$
Z)$: a contradiction.
\medskip
\end{proof}
A {\em defect of forward confluence of a regular clip $({\mathcal T},{\ll},{\triangleright})$}\/ is a triple $((i,s,\alpha,X),
$\linebreak$
(j,t,\beta,Y),(k,u,\gamma,Z))$ where $(i,s,\alpha,X),(j,t,\beta,Y),(k,u,\gamma,Z){\in}{\mathcal T}$ are such that
\begin{itemize}
\item $(j,t,\beta,Y){\ll}(i,s,\alpha,X)$,
\item $(j,t,\beta,Y){\triangleright}(k,u,\gamma,Z)$,
\item for all $(l,v,\delta,T){\in}{\mathcal T}$, either $(i,s,\alpha,X){\not\triangleright}(l,v,\delta,T)$, or $(k,u,\gamma,Z){\not\ll}(l,v,\delta,
$\linebreak$
T)$.
\end{itemize}
The {\em rank of the defect $((i,s,\alpha,X),(j,t,\beta,Y),(k,u,\gamma,Z))$ of forward confluence}\/ is $\alpha$.
\\
\\
The {\em height of the defect $((i,s,\alpha,X),(j,t,\beta,Y),(k,u,\gamma,Z))$ of forward confluence}\/ is $X$.
\\
\\
For all $\alpha{\in}\N$, we say that a regular clip $({\mathcal T},{\ll},{\triangleright})$ is
\begin{itemize}
\item {\em $\alpha$-clean for maximality}\/ if for all $\beta{\in}\N$, if $\beta{<}\alpha$ then $({\mathcal T},{\ll},{\triangleright})$ contains no defect of maximality of rank $\beta$,
\item {\em $\alpha$-clean for accessibility}\/ if for all $\beta{\in}\N$, if $\beta{<}\alpha$ then $({\mathcal T},{\ll},{\triangleright})$ contains no defect of accessibility of rank $\beta$,
\item {\em $\alpha$-clean for downward confluence}\/ if for all $\beta{\in}\N$, if $\beta{<}\alpha$ then $({\mathcal T},{\ll},{\triangleright})$ contains no defect of downward confluence of rank $\beta$,
\item {\em $\alpha$-clean for forward confluence}\/ if for all $\beta{\in}\N$, if $\beta{<}\alpha$ then $({\mathcal T},{\ll},{\triangleright})$ contains no defect of forward confluence of rank $\beta$.
\end{itemize}
We say that a regular clip is {\em clean}\/ if it contains no defect.
\begin{lemma}\label{lemma:about:the:cleanness:of:initial:clips}
The initial clip of $s_{0}$ is $0$-clean for maximality, $0$-clean for accessibility, $0$-clean for downward confluence and $0$-clean for forward confluence.
\end{lemma}
Defects are not welcome.
Luckily, they are repairable.
In Sections~\ref{section:reparations:maximality:defects}, \ref{section:reparations:accessibility:defects} and~\ref{section:reparations:confluence:defects}, we are showing how defects can be repaired.
\section{Repair of maximality defects}\label{section:reparations:maximality:defects}
In this section, we are showing how maximality defects can be repaired.
\\
\\
The {\em repair of a defect $((i,s,\alpha,X),B,C)$ of maximality of a regular clip $({\mathcal T},{\ll},{\triangleright})$}\/ consists in sequentially executing the following actions:
\begin{itemize}
\item add a tip $(j,t,\beta,Y)$ to ${\mathcal T}$ such that $j$ is new, $s{<_{c}}t$, $t{\not\models}B{\rightarrow}C$, $t$ is maximal with respect to $B{\rightarrow}C$, $\beta{=}\alpha$ and $Y{=}X{+}1$,\footnote{The existence of $t{\in}W_{c}$ such that $s{<_{c}}t$, $t{\not\models}B{\rightarrow}C$ and $t$ is maximal with respect to $B{\rightarrow}C$ is an immediate consequence of Lemma~\ref{lemma:maximal} and the fact that $((i,s,\alpha,X),B,C)$ is a defect of maximality.}
\item add the couple $((i,s,\alpha,X),(j,t,\beta,Y))$ to $\ll$.
\end{itemize}
Obviously, the resulting clip is coherent.
\\
\\
Moreover, since the resulting clip is obtained by adding the tip $(j,t,\beta,Y)$ to ${\mathcal T}$ and the couple $((i,s,\alpha,X),(j,t,\beta,Y))$ to $\ll$, then the resulting clip is regular.
\\
\\
In other respect, notice that this repair is the repair of a defect of rank $\alpha$ and height $X$ that only introduces in ${\mathcal T}$ a tip of rank $\alpha$ and height $X{+}1$.
\\
\\
As well, notice that $\degre(j,t,\beta,Y){<}\degre(i,s,\alpha,X)$.
\\
\\
Given $\alpha{\in}\N$ and a regular clip $({\mathcal T},{\ll},{\triangleright})$, the {\em maximality procedure}\/ is defined as follows:
\begin{enumerate}
\item $x:=({\mathcal T},{\ll},{\triangleright})$,
\item $X:=0$,
\item while $x$ contains defects of maximality of rank $\alpha$ do
\begin{enumerate}
\item repair in $x$ all defects of maximality of rank $\alpha$ and height $X$,
\item $X:=X{+}1$.
\end{enumerate}
\end{enumerate}
Given $\alpha{\in}\N$ and a regular clip $({\mathcal T},{\ll},{\triangleright})$, the role of the maximality procedure is to iteratively repair all defects of maximality of $({\mathcal T},{\ll},{\triangleright})$ of rank $\alpha$.
\begin{lemma}\label{lemma:about:height:maximal:in:tips:repairing:max:defects}
Given $\alpha{\in}\N$ and a regular clip $({\mathcal T},{\ll},{\triangleright})$, at any moment of the execution of the maximality procedure, for all tips $(j,t,\beta,Y)$ occurring in $x$, $Y{\leq}\max\{Z{\in}\N$: $(k,u,\gamma,Z){\in}{\mathcal T}\}{+}\card(\Sigma_{A})$.
\end{lemma}
\begin{proof}
It suffices to notice that for all $X{\in}\N$, the repair of a defect of maximality of rank $\alpha$ and height $X$ only introduces a tip of rank $\alpha$ and height $X{+}1$.
Moreover, as noticed above, the degree of the introduced tip is strictly smaller than the degree of the tip that has caused this introduction.
Hence, by Lemma~\ref{lemma:about:degree:and:Sigma:A}, at any moment of the execution of the maximality procedure, for all tips $(j,t,\beta,Y)$ occurring in $x$, $Y{\leq}\max\{Z{\in}\N$: $(k,u,\gamma,Z){\in}{\mathcal T}\}{+}\card(\Sigma_{A})$.
\medskip
\end{proof}
\begin{lemma}\label{lemma:procedure:defects:maximality:terminates}
Given $\alpha{\in}\N$ and a regular clip $({\mathcal T},{\ll},{\triangleright})$, the maximality procedure terminates.
\end{lemma}
\begin{proof}
By Lemma~\ref{lemma:about:height:maximal:in:tips:repairing:max:defects}.
\medskip
\end{proof}
\begin{lemma}\label{lemma:about:alpha:clean:coherent:clips:repairing:maximality}
Let $\alpha{\in}\N$ and $({\mathcal T},{\ll},{\triangleright})$ be a regular clip.
If $({\mathcal T},{\ll},{\triangleright})$ is $\alpha$-clean for maximality, $\alpha$-clean for accessibility, $\alpha$-clean for downward confluence and $\alpha$-clean for forward confluence then the regular clip obtained from $({\mathcal T},{\ll},{\triangleright})$ after the execution of the maximality procedure is $(\alpha{+}1)$-clean for maximality, $\alpha$-clean for accessibility, $\alpha$-clean for downward confluence and $\alpha$-clean for forward confluence.
\end{lemma}
\begin{proof}
Suppose $({\mathcal T},{\ll},{\triangleright})$ is $\alpha$-clean for maximality, $\alpha$-clean for accessibility, $\alpha$-clean for downward confluenceand $\alpha$-clean for forward confluence.
Notice that by Lemma~\ref{lemma:procedure:defects:maximality:terminates}, the maximality procedure terminates.
Since the execution of the maximality procedure only introduces tips of rank $\alpha$, then the regular clip obtained from $({\mathcal T},{\ll},{\triangleright})$ after the execution of the maximality procedure is $(\alpha{+}1)$-clean for maximality, $\alpha$-clean for accessibility, $\alpha$-clean for downward confluence and $\alpha$-clean for forward confluence.
\medskip
\end{proof}
\section{Repair of accessibility defects}\label{section:reparations:accessibility:defects}
In this section, we are showing how defects of $\square$-accessibility and defects of $\lozenge$-accessibility can be repaired.
\\
\\
The {\em repair of a defect $((i,s,\alpha,X),B)$ of $\square$-accessibility of a regular clip $({\mathcal T},{\ll},{\triangleright})$}\/ consists in sequentially executing the following actions:
\begin{itemize}
\item add a tip $(j,t,\beta,Y)$ to ${\mathcal T}$ such that $j$ is new, $t{\in}R_{c}(s)$, $t{\not\models}B$, $\beta{=}\alpha{+}1$ and $Y{=}X$,\footnote{The existence of $t{\in}W_{c}$ such that $t{\in}R_{c}(s)$ and $t{\not\models}B$ is an immediate consequence of the fact that $((i,s,\alpha,X),B)$ is a defect of $\square$-accessibility.}
\item add the couple $((i,s,\alpha,X),(j,t,\beta,Y))$ to $\triangleright$.
\end{itemize}
Obviously, the resulting clip is coherent.
\\
\\
Moreover, since the resulting clip is obtained by adding the tip $(j,t,\beta,Y)$ to ${\mathcal T}$ and the couple $((i,s,\alpha,X),(j,t,\beta,Y))$ to $\triangleright$, then the resulting clip is regular.
\\
\\
In other respect, notice that this repair is the repair of a defect of rank $\alpha$ and height $X$ that only introduces in ${\mathcal T}$ a tip of rank $\alpha{+}1$ and height $X$.
\\
\\
The {\em repair of a defect $((i,s,\alpha,X),B)$ of $\lozenge$-accessibility of a regular clip $({\mathcal T},{\ll},{\triangleright})$}\/ consists in sequentially executing the following actions:
\begin{itemize}
\item add a tip $(j,t,\beta,Y)$ to ${\mathcal T}$ such that $j$ is new, $t{\in}R_{c}(s)$, $t{\models}B$, $\beta{=}\alpha{+}1$ and $Y{=}X$,\footnote{The existence of $t{\in}W_{c}$ such that $t{\in}R_{c}(s)$ and $t{\models}B$ is an immediate consequence of the fact that $((i,s,\alpha,X),B)$ is a defect of $\lozenge$-accessibility.}
\item add the couple $((i,s,\alpha,X),(j,t,\beta,Y))$ to $\triangleright$.
\end{itemize}
Obviously, the resulting clip is coherent.
\\
\\
Moreover, since the resulting clip is obtained by adding the tip $(j,t,\beta,Y)$ to ${\mathcal T}$ and the couple $((i,s,\alpha,X),(j,t,\beta,Y))$ to $\triangleright$, then the resulting clip is regular.
\\
\\
In other respect, notice that this repair is the repair of a defect of rank $\alpha$ and height $X$ that only introduces in ${\mathcal T}$ a tip of rank $\alpha{+}1$ and height $X$.
\\
\\
Given $\alpha{\in}\N$ and a regular clip $({\mathcal T},{\ll},{\triangleright})$, the {\em accessibility procedure}\/ is defined as follows:
\begin{enumerate}
\item $x:=({\mathcal T},{\ll},{\triangleright})$,
\item $X:=0$,
\item while $x$ contains defects of accessibility of rank $\alpha$ do
\begin{enumerate}
\item repair in $x$ all defects of accessibility of rank $\alpha$ and height $X$,
\item $X:=X{+}1$.
\end{enumerate}
\end{enumerate}
Given $\alpha{\in}\N$ and a regular clip $({\mathcal T},{\ll},{\triangleright})$, the role of the accessibility procedure is to iteratively repair all defects of accessibility of $({\mathcal T},{\ll},{\triangleright})$ of rank $\alpha$.
\begin{lemma}\label{lemma:about:height:maximal:in:tips:repairing:accessibility:defects}
Given $\alpha{\in}\N$ and a regular clip $({\mathcal T},{\ll},{\triangleright})$, at any moment of the execution of the accessibility procedure, for all tips $(j,t,\beta,Y)$ occurring in $x$, $Y{\leq}\max\{Z{\in}\N$: $(k,u,\gamma,Z){\in}{\mathcal T}\}$.
\end{lemma}
\begin{proof}
It suffices to notice that for all $X{\in}\N$, the repair of a defect of accessibility of rank $\alpha$ and height $X$ only introduces a tip of rank $\alpha{+}1$ and height $X$.
\medskip
\end{proof}
\begin{lemma}\label{lemma:procedure:defects:accessibility:terminates}
Given $\alpha{\in}\N$ and a regular clip $({\mathcal T},{\ll},{\triangleright})$, the accessibility procedure terminates.
\end{lemma}
\begin{proof}
By Lemma~\ref{lemma:about:height:maximal:in:tips:repairing:accessibility:defects}.
\medskip
\end{proof}
\begin{lemma}\label{lemma:about:alpha:clean:coherent:clips:repairing:accessibility}
Let $\alpha{\in}\N$ and $({\mathcal T},{\ll},{\triangleright})$ be a regular clip.
If $({\mathcal T},{\ll},
$\linebreak$
{\triangleright})$ is $(\alpha{+}1)$-clean for maximality, $\alpha$-clean for accessibility, $\alpha$-clean for downward confluence and $\alpha$-clean for forward confluence then the regular clip obtained from $({\mathcal T},{\ll},{\triangleright})$ after the execution of the accessibility procedure is $(\alpha{+}1)$-clean for maximality, $(\alpha{+}1)$-clean for accessibility, $\alpha$-clean for downward confluence and $\alpha$-clean for forward confluence.
\end{lemma}
\begin{proof}
Suppose $({\mathcal T},{\ll},{\triangleright})$ is $(\alpha{+}1)$-clean for maximality, $\alpha$-clean for accessibility, $\alpha$-clean for downward confluence and $\alpha$-clean for forward confluence.
Notice that by Lemma~\ref{lemma:procedure:defects:accessibility:terminates}, the accessibility procedure terminates.
Since the execution of the accessibility procedure only introduces tips of rank $\alpha{+}1$, then the regular clip obtained from $({\mathcal T},{\ll},{\triangleright})$ after the execution of the accessibility procedure is $(\alpha{+}1)$-clean for maximality, $(\alpha{+}1)$-clean for accessibility, $\alpha$-clean for downward confluence and $\alpha$-clean for forward confluence.
\medskip
\end{proof}
\section{Repair of confluence defects}\label{section:reparations:confluence:defects}
In this section, we are showing how defects of downward confluence and defects of forward confluence can be repaired.
\\
\\
The {\em repair of a defect $((i,s,\alpha,X),(j,t,\beta,Y),(k,u,\gamma,Z))$ of downward confluence of a regular clip $({\mathcal T},{\ll},{\triangleright})$}\/ consists in sequentially executing the following actions:
\begin{itemize}
\item add a tip $(l,v,\delta,T)$ to ${\mathcal T}$ such that $l$ is new, $v{\in}R_{c}(s)$, $v{\leq_{c}}u$, $\delta{=}\alpha{+}1$, $\delta{=}\gamma$, $T{=}X$ and $Z{=}T{+}1$,\footnote{The existence of $v{\in}W_{c}$ such that $v{\in}R_{c}(s)$ and $v{\leq_{c}}u$ is an immediate consequence of the fact that $((i,s,\alpha,X),(j,t,\beta,Y),(k,u,\gamma,Z))$ is a defect of downward confluence.}
\item add the couple $((i,s,\alpha,X),(l,v,\delta,T))$ to $\triangleright$,
\item add the couple $((l,v,\delta,T),(k,u,\gamma,Z))$ to $\ll$.
\end{itemize}
Obviously, the resulting clip is coherent.
\\
\\
Moreover, since the resulting clip is obtained by adding the tip $(l,v,\delta,T)$ to ${\mathcal T}$, the couple $((i,s,\alpha,X),(l,v,\delta,T))$ to $\triangleright$ and the couple $((l,v,\delta,T),(k,u,\gamma,Z))$ to $\ll$, then by Lemma~\ref{lemma:about:defects:of:downward:confluence}, the resulting clip is regular.
\\
\\
In other respect, notice that this repair is the repair of a defect of rank $\alpha$ and height $X$ that only introduces in ${\mathcal T}$ a tip of rank $\alpha{+}1$ and height $X$.
\\
\\
Given $\alpha{\in}\N$ and a regular clip $({\mathcal T},{\ll},{\triangleright})$, the {\em procedure of downward confluence}\/ is defined as follows:
\begin{enumerate}
\item $x:=({\mathcal T},{\ll},{\triangleright})$,
\item $X:=\hauteur_{\alpha}({\mathcal T})$,
\item while $x$ contains defects of downward confluence of rank $\alpha$ do
\begin{enumerate}
\item repair in $x$ all defects of downward confluence of rank $\alpha$ and height $X$,
\item $X:=X{-}1$.
\end{enumerate}
\end{enumerate}
Given $\alpha{\in}\N$ and a regular clip $({\mathcal T},{\ll},{\triangleright})$, the role of the procedure of downward confluence is to iteratively repair all defects of downward confluence of $({\mathcal T},{\ll},{\triangleright})$ of rank $\alpha$.
\begin{lemma}\label{lemma:about:height:maximal:in:tips:repairing:downward:confluence:defects}
Given $\alpha{\in}\N$ and a regular clip $({\mathcal T},{\ll},{\triangleright})$, at any moment of the execution of the procedure of downward confluence, for all tips $(j,t,\beta,Y)$ occurring in $x$, $Y{\leq}\max\{Z{\in}\N$: $(k,u,\gamma,Z){\in}{\mathcal T}\}$.
\end{lemma}
\begin{proof}
It suffices to notice that for all $X{\in}\N$, the repair of a defect of downward confluence of rank $\alpha$ and height $X$ only introduces a tip of rank $\alpha{+}1$ and height $X$.
\medskip
\end{proof}
\begin{lemma}\label{lemma:procedure:defects:downward:confluence:terminates}
Given $\alpha{\in}\N$ and a regular clip $({\mathcal T},{\ll},{\triangleright})$, the procedure of downward confluence terminates.
\end{lemma}
\begin{proof}
By Lemma~\ref{lemma:about:height:maximal:in:tips:repairing:downward:confluence:defects}.
\medskip
\end{proof}
\begin{lemma}\label{lemma:about:alpha:clean:coherent:clips:repairing:downward:confluence}
Let $\alpha{\in}\N$ and $({\mathcal T},{\ll},{\triangleright})$ be a regular clip.
If $({\mathcal T},{\ll},{\triangleright})$ is $(\alpha{+}1)$-clean for maximality, $(\alpha{+}1)$-clean for accessibility, $\alpha$-clean for downward confluence and $\alpha$-clean for forward confluence then the regular clip obtained from $({\mathcal T},{\ll},{\triangleright})$ after the execution of the procedure of downward confluence is $(\alpha{+}1)$-clean for maximality, $(\alpha{+}1)$-clean for accessibility, $(\alpha{+}1)$-clean for downward confluence and $\alpha$-clean for forward confluence.
\end{lemma}
\begin{proof}
Suppose $({\mathcal T},{\ll},{\triangleright})$ is $(\alpha{+}1)$-clean for maximality, $(\alpha{+}1)$-clean for accessibility, $\alpha$-clean for downward confluence and $\alpha$-clean for forward confluence.
Notice that by Lemma~\ref{lemma:procedure:defects:downward:confluence:terminates}, the procedure of downward confluence terminates.
Since the execution of the procedure of downward confluence only introduces tips of rank $\alpha{+}1$, then the regular clip obtained from $({\mathcal T},{\ll},{\triangleright})$ after the execution of the procedure of downward confluence is $(\alpha{+}1)$-clean for maximality, $(\alpha{+}1)$-clean for accessibility, $(\alpha{+}1)$-clean for downward confluence and $\alpha$-clean for forward confluence.
\medskip
\end{proof}
The {\em repair of a defect $((i,s,\alpha,X),(j,t,\beta,Y),(k,u,\gamma,Z))$ of forward confluence of a regular clip $({\mathcal T},{\ll},{\triangleright})$}\/ consists in sequentially executing the following actions:
\begin{itemize}
\item add a tip $(l,v,\delta,T)$ to ${\mathcal T}$ such that $l$ is new, $v{\in}R_{c}(s)$, $u{\leq_{c}}v$, $\delta{=}\alpha{+}1$, $\delta{=}\gamma$, $T{=}X$ and $T{=}Z{+}1$,\footnote{The existence of $v{\in}W_{c}$ such that $v{\in}R_{c}(s)$ and $u{\leq_{c}}v$ is an immediate consequence of the fact that $((i,s,\alpha,X),(j,t,\beta,Y),(k,u,\gamma,Z))$ is a defect of forward confluence.}
\item add the couple $((i,s,\alpha,X),(l,v,\delta,T))$ to $\triangleright$,
\item add the couple $((k,u,\gamma,Z),(l,v,\delta,T))$ to $\ll$.
\end{itemize}
Obviously, the resulting clip is coherent.
\\
\\
Moreover, since the resulting clip is obtained by adding the tip $(l,v,\delta,T)$ to ${\mathcal T}$, the couple $((i,s,\alpha,X),(l,v,\delta,T))$ to $\triangleright$ and the couple $((k,u,\gamma,Z),(l,v,\delta,T))$ to $\ll$, then the resulting clip is regular.
\\
\\
In other respect, notice that this repair is the repair of a defect of rank $\alpha$ and height $X$ that only introduces in ${\mathcal T}$ a tip of rank $\alpha{+}1$ and height $X$.
\\
\\
Given $\alpha{\in}\N$ and a regular clip $({\mathcal T},{\ll},{\triangleright})$, the {\em procedure of forward confluence}\/ is defined as follows:
\begin{enumerate}
\item $x:=({\mathcal T},{\ll},{\triangleright})$,
\item $X:=0$,
\item while $x$ contains defects of forward confluence of rank $\alpha$ do
\begin{enumerate}
\item repair in $x$ all defects of forward confluence of rank $\alpha$ and height $X$,
\item $X:=X{+}1$.
\end{enumerate}
\end{enumerate}
Given $\alpha{\in}\N$ and a regular clip $({\mathcal T},{\ll},{\triangleright})$, the role of the procedure of forward confluence is to iteratively repair all defects of forward confluence of $({\mathcal T},{\ll},{\triangleright})$ of rank $\alpha$.
\begin{lemma}\label{lemma:about:height:maximal:in:tips:repairing:forward:confluence:defects}
Given $\alpha{\in}\N$ and a regular clip $({\mathcal T},{\ll},{\triangleright})$, at any moment of the execution of the procedure of forward confluence, for all tips $(j,t,\beta,Y)$ occurring in $x$, $Y{\leq}\max\{Z{\in}\N$: $(k,u,\gamma,Z){\in}{\mathcal T}\}$.
\end{lemma}
\begin{proof}
It suffices to notice that for all $X{\in}\N$, the repair of a defect of forward confluence of rank $\alpha$ and height $X$ only introduces a tip of rank $\alpha{+}1$ and height $X$.
\medskip
\end{proof}
\begin{lemma}\label{lemma:procedure:defects:forward:confluence:terminates}
Given $\alpha{\in}\N$ and a regular clip $({\mathcal T},{\ll},{\triangleright})$, the procedure of forward confluence terminates.
\end{lemma}
\begin{proof}
By Lemma~\ref{lemma:about:height:maximal:in:tips:repairing:forward:confluence:defects}.
\medskip
\end{proof}
\begin{lemma}\label{lemma:about:alpha:clean:coherent:clips:repairing:forward:confluence}
Let $\alpha{\in}\N$ and $({\mathcal T},{\ll},{\triangleright})$ be a regular clip.
If $({\mathcal T},{\ll},{\triangleright})$ is $(\alpha{+}1)$-clean for maximality, $(\alpha{+}1)$-clean for accessibility, $(\alpha{+}1)$-clean for downward confluence and $\alpha$-clean for forward confluence then the regular clip obtained from $({\mathcal T},{\ll},{\triangleright})$ after the execution of the procedure of forward confluence is $(\alpha{+}1)$-clean for maximality, $(\alpha{+}1)$-clean for accessibility, $(\alpha{+}1)$-clean for downward confluence and $(\alpha{+}1)$-clean for forward confluence.
\end{lemma}
\begin{proof}
Suppose $({\mathcal T},{\ll},{\triangleright})$ is $(\alpha{+}1)$-clean for maximality, $(\alpha{+}1)$-clean for accessibility, $(\alpha{+}1)$-clean for downward confluence and $\alpha$-clean for forward confluence.
Notice that by Lemma~\ref{lemma:procedure:defects:forward:confluence:terminates}, the procedure of forward confluence terminates.
Since the execution of the procedure of forward confluence only introduces tips of rank $\alpha{+}1$, then the regular clip obtained from $({\mathcal T},{\ll},{\triangleright})$ after the execution of the procedure of forward confluence is $(\alpha{+}1)$-clean for maximality, $(\alpha{+}1)$-clean for accessibility, $(\alpha{+}1)$-clean for downward confluence and $(\alpha{+}1)$-clean for forward confluence.
\medskip
\end{proof}
\section{A decision procedure}\label{section:a:decision:procedure}
The {\em saturation procedure}\/ is the following sequence of steps:
\begin{description}
\item[$(1)$] $x:=(\{(0,s_{0},0,0)\},\emptyset,\emptyset)$,
\item[$(2)$] $\alpha:=0$,
\item[$(3)$] repair in $x$ all defects of maximality of rank $\alpha$ by applying the maximality procedure,
\item[$(4)$] if the finitely indexed family of $P_{A}$-labelled trees associated to $x$ and $\alpha$ is dreary then goto $(10)$ else goto $(5)$,
\item[$(5)$] repair in $x$ all defects of accessibility of rank $\alpha$ by applying the accessibility procedure,
\item[$(6)$] repair in $x$ all defects of downward confluence of rank $\alpha$ by applying the procedure of downward confluence,
\item[$(7)$] repair in $x$ all defects of forward confluence of rank $\alpha$ by applying the procedure of forward confluence,
\item[$(8)$] $\alpha:=\alpha{+}1$,
\item[$(9)$] goto $(3)$,
\item[$(10)$] halt.
\end{description}
The role of the saturation procedure is to iteratively repair all defects of $(\{(0,s_{0},
$\linebreak$
0,0)\},\emptyset,\emptyset)$.
A few remarks must be made about its steps.
\\
\\
About step~$(1)$, by Lemma~\ref{lemma:about:the:cleanness:of:initial:clips}, the regular clip $x$ is $0$-clean for maximality, $0$-clean for accessibility, $0$-clean for downward confluence and $0$-clean for forward confluence.
\\
\\
Therefore, about step~$(2)$, the regular clip $x$ is $\alpha$-clean for maximality, $\alpha$-clean for accessibility, $\alpha$-clean for downward confluence and $\alpha$-clean for forward confluence.
\\
\\
As a result, by Lemma~\ref{lemma:procedure:defects:maximality:terminates}, the procedure used in step~$(3)$ terminates.
Moreover, since after step~$(2)$, the regular clip $x$ given as an input to the maximality procedure is $\alpha$-clean for maximality, $\alpha$-clean for accessibility, $\alpha$-clean for downward confluence and $\alpha$-clean for forward confluence, then by Lemma~\ref{lemma:about:alpha:clean:coherent:clips:repairing:maximality}, the regular clip $x$ given as an output to the procedure used in step~$(3)$ is $(\alpha{+}1)$-clean for maximality, $\alpha$-clean for accessibility, $\alpha$-clean for downward confluence and $\alpha$-clean for forward confluence.
\\
\\
Concerning step~$(4)$, obviously, for all $\beta{\in}\N$, if $\alpha{<}\beta$ then the $\beta$-slice of $x$ is empty.
\\
\\
By Lemma~\ref{lemma:procedure:defects:accessibility:terminates}, the procedure used in step~$(5)$ terminates.
Moreover, since after step~$(3)$, the regular clip $x$ given as an input to the procedure used in step~$(5)$ is $(\alpha{+}1)$-clean for maximality, $\alpha$-clean for accessibility, $\alpha$-clean for downward confluence and $\alpha$-clean for forward confluence, then by Lemma~\ref{lemma:about:alpha:clean:coherent:clips:repairing:accessibility}, the regular clip $x$ given as an output to the procedure used in step~$(5)$ is $(\alpha{+}1)$-clean for maximality, $(\alpha{+}1)$-clean for accessibility, $\alpha$-clean for downward confluence and $\alpha$-clean for forward confluence.
\\
\\
By Lemma~\ref{lemma:procedure:defects:downward:confluence:terminates}, the procedure used in step~$(6)$ terminates.
Moreover, since after step~$(5)$, the regular clip $x$ given as an input to the procedure used in step~$(6)$ is $(\alpha{+}1)$-clean for maximality, $(\alpha{+}1)$-clean for accessibility, $\alpha$-clean for downward confluence and $\alpha$-clean for forward confluence, then by Lemma~\ref{lemma:about:alpha:clean:coherent:clips:repairing:downward:confluence}, the regular clip $x$ given as an output to the procedure used in step~$(6)$ is $(\alpha{+}1)$-clean for maximality, $(\alpha{+}1)$-clean for accessibility, $(\alpha{+}1)$-clean for downward confluence and $\alpha$-clean for forward confluence.
\\
\\
By Lemma~\ref{lemma:procedure:defects:forward:confluence:terminates}, the procedure used in step~$(7)$ terminates.
Moreover, since after step~$(6)$, the regular clip $x$ given as an input to the procedure used in step~$(7)$ is $(\alpha{+}1)$-clean for maximality, $(\alpha{+}1)$-clean for accessibility, $(\alpha{+}1)$-clean for downward confluence and $\alpha$-clean for forward confluence, then by Lemma~\ref{lemma:about:alpha:clean:coherent:clips:repairing:forward:confluence}, the regular clip $x$ given as an output to the procedure used in step~$(7)$ is $(\alpha{+}1)$-clean for maximality, $(\alpha{+}1)$-clean for accessibility, $(\alpha{+}1)$-clean for downward confluence and $(\alpha{+}1)$-clean for forward confluence.
\\
\\
With respect to steps~$(8)$ and~$(9)$, there is nothing to say.
\\
\\
And in the end, when the saturation procedure arrives at step~$(10)$, the finitely indexed family of $P_{A}$-labelled trees associated to $x$ and $\alpha$ is dreary.
\begin{lemma}\label{proposition:principal}
The saturation procedure terminates.
\end{lemma}
\begin{proof}
By Lemma~\ref{proposition:about:labelled:trees:to:be:dreary}.
\medskip
\end{proof}
Let $({\mathcal T},{\ll},{\triangleright})$ and $\alpha_{f}$ be the values of $x$ and $\alpha$ when the saturation procedure terminates.
\\
\\
Therefore, for all $\beta{\in}\N$, if $\alpha_{f}{<}\beta$ then the $\beta$-slice of $({\mathcal T},{\ll},{\triangleright})$ is empty.\footnote{See the above remark about step~$(4)$.}
\\
\\
Moreover, $({\mathcal T},{\ll},{\triangleright})$ is $(\alpha_{f}{+}1)$-clean for maximality, $\alpha_{f}$-clean for accessibility, $\alpha_{f}$-clean for downward confluence and $\alpha_{f}$-clean for forward confluence.\footnote{See the above remarks about steps~$(3)$, $(5)$, $(6)$ and~$(7)$.}
\\
\\
In other respect, the finitely indexed family $({\mathcal T}_{\beta},{\ll_{\beta}},\lambda_{\beta})_{\beta{\in}(\alpha_{f})}$ of $P_{A}$-labelled trees associated to $({\mathcal T},{\ll},{\triangleright})$ and $\alpha_{f}$ is dreary.\footnote{See the above remark about step~$(10)$.}
As a result, there exists $\beta_{f}{\in}(\alpha_{f})$ such that $\beta_{f}{<}\alpha_{f}$ and $({\mathcal T}_{\beta_{f}},{\ll_{\beta_{f}}},\lambda_{\beta_{f}}){\sim}({\mathcal T}_{\alpha_{f}},{\ll_{\alpha_{f}}},\lambda_{\alpha_{f}})$.
Therefore, there exists a function $f$: ${\mathcal T}_{\alpha_{f}}\longrightarrow{\mathcal T}_{\beta_{f}}$ such that
\begin{itemize}
\item for all $(i,s,\alpha_{f},X),(j,t,\alpha_{f},Y){\in}{\mathcal T}_{\alpha_{f}}$, if $(i,s,\alpha_{f},X){\ll_{\alpha_{f}}}(j,t,\alpha_{f},Y)$ then $f(i,s,\alpha_{f},X){\ll_{\beta_{f}}^{\star}}f(j,t,\alpha_{f},Y)$,
\item for all $(k,u,\alpha_{f},Z){\in}{\mathcal T}_{\alpha_{f}}$, $\lambda_{\alpha_{f}}(k,u,\alpha_{f},Z){=}\lambda_{\beta_{f}}(f(k,u,\alpha_{f},Z))$.
\end{itemize}
\section{Finite frame property}\label{section:finite:model:property}
Let ${\triangleright^{\prime}}{=}\{((i,s,\alpha_{f},X),(l,v,\beta_{f}{+}1,T))$: $(i,s,\alpha_{f},X){\in}{\mathcal T}_{\alpha_{f}}$, $(l,v,\beta_{f}{+}1,T){\in}{\mathcal T}_{\beta_{f}{+}1}$ and $f(i,s,\alpha_{f},X){\triangleright}(l,v,\beta_{f}{+}1,T)\}$.
Let $(W^{\prime},{\leq^{\prime}},{R^{\prime}})$ be the frame such that
\begin{itemize}
\item $W^{\prime}{=}{\mathcal T}$,\footnote{Obviously, $W^{\prime}$ is nonempty.}
\item ${\leq^{\prime}}{=}{\ll^{\star}}$,\footnote{Obviously, $\leq^{\prime}$ is a preorder on $W^{\prime}$.}
\item ${R^{\prime}}{=}({\triangleright}{\cup}{\triangleright^{\prime}})^{+}$.
\end{itemize}
\begin{lemma}\label{lemma:frame:prime:is:appropriate}
The frame $(W^{\prime},{\leq^{\prime}},{R^{\prime}})$ is finite, transitive, downward confluent and forward confluent.
\end{lemma}
\begin{proof}
About downward confluence, it suffices to demonstrate that for all $\gamma{\in}\N$, if $\gamma{<}\alpha_{f}$ then ${\ll_{\gamma}}{\circ}{\triangleright}{\subseteq}{\triangleright}{\circ}{\ll_{\gamma{+}1}}$ and ${\ll_{\alpha_{f}}}{\circ}{\triangleright^{\prime}}{\subseteq}{\triangleright^{\prime}}{\circ}{\ll_{\beta_{f}{+}1}^{\star}}$.
The former inclusions are consequences of the effect of step~$(6)$ in the saturation procedure.
As for the latter inclusion, suppose $(i,s,\alpha_{f},X),(j,t,\alpha_{f},Y){\in}{\mathcal T}_{\alpha_{f}}$ and $(l,v,\beta_{f}{+}1,T){\in}{\mathcal T}_{\beta_{f}{+}1}$ are such that $(i,s,\alpha_{f},X){\ll_{\alpha_{f}}}(j,t,\alpha_{f},Y)$ and $(j,t,\alpha_{f},Y){\triangleright^{\prime}}(l,v,\beta_{f}{+}1,T)$.
\linebreak
Hence, $f(i,s,\alpha_{f},X){\ll_{\beta_{f}}^{\star}}f(j,t,\alpha_{f},Y)$ and $f(j,t,\alpha_{f},Y){\triangleright}(l,v,\beta_{f}{+}1,T)$.
Thus, as a consequence of the effect of step~$(6)$ in the saturation procedure, there exists $(m,w,\beta_{f}{+}1,U){\in}{\mathcal T}_{\beta_{f}{+}1}$ such that $f(i,s,\alpha_{f},X){\triangleright}(m,w,\beta_{f}{+}1,U)$ and $(m,w,\beta_{f}{+}
$\linebreak$
1,U){\ll_{\beta_{f}{+}1}^{\star}}(l,v,\beta_{f}{+}1,T)$.
Consequently, $(i,s,\alpha_{f},X){\triangleright^{\prime}}(m,w,\beta_{f}{+}1,U)$ and $(m,
$\linebreak$
w,\beta_{f}{+}1,U){\ll_{\beta_{f}{+}1}^{\star}}(l,v,\beta_{f}{+}1,T)$.
\\
\\
As for forward confluence, it suffices to demonstrate that for all $\gamma{\in}\N$, if $\gamma{<}\alpha_{f}$ then ${\gg_{\gamma}}{\circ}{\triangleright}{\subseteq}{\triangleright}{\circ}{\gg_{\gamma{+}1}}$ and ${\gg_{\alpha_{f}}}{\circ}{\triangleright^{\prime}}{\subseteq}{\triangleright^{\prime}}{\circ}{\gg_{\beta_{f}{+}1}^{\star}}$.
The former inclusions are consequences of the effect of step~$(7)$ in the saturation procedure.
As for the latter inclusion, suppose $(i,s,\alpha_{f},X),(j,t,\alpha_{f},Y){\in}{\mathcal T}_{\alpha_{f}}$ and $(l,v,\beta_{f}{+}1,T){\in}{\mathcal T}_{\beta_{f}{+}1}$ are such that $(i,s,\alpha_{f},X){\gg_{\alpha_{f}}}(j,t,\alpha_{f},Y)$ and $(j,t,\alpha_{f},Y){\triangleright^{\prime}}(l,v,\beta_{f}{+}1,T)$.
Hence, $f(i,s,\alpha_{f},X){\gg_{\beta_{f}}^{\star}}f(j,t,\alpha_{f},Y)$ and $f(j,t,\alpha_{f},Y){\triangleright}(l,v,\beta_{f}{+}1,T)$.
Thus, as a consequence of the effect of step~$(7)$ in the saturation procedure, there exists $(m,w,\beta_{f}{+}1,U){\in}{\mathcal T}_{\beta_{f}{+}1}$ such that $f(i,s,\alpha_{f},X){\triangleright}(m,w,\beta_{f}{+}1,U)$ and $(m,w,\beta_{f}{+}
$\linebreak$
1,U){\gg_{\beta_{f}{+}1}^{\star}}(l,v,\beta_{f}{+}1,T)$.
Consequently, $(i,s,\alpha_{f},X){\triangleright^{\prime}}(m,w,\beta_{f}{+}1,U)$ and $(m,
$\linebreak$
w,\beta_{f}{+}1,U){\gg_{\beta_{f}{+}1}^{\star}}(l,v,\beta_{f}{+}1,T)$.
\medskip
\end{proof}
The frame $(W^{\prime},{\leq^{\prime}},{R^{\prime}})$ is called {\em saturated frame of $s_{0}$.}
\\
\\
The {\em saturated valuation of $s_{0}$}\/ is the valuation $V^{\prime}$: $\At\longrightarrow\wp(W^{\prime})$ on $(W^{\prime},{\leq^{\prime}},{R^{\prime}})$ such that for all $p{\in}\At$, $V^{\prime}(p){=}\{(l,v,\delta,T){\in}{\mathcal T}$: $v{\in}V_{c}(p)\}$.\footnote{Obviously, for all $p{\in}\At$, $V^{\prime}(p)$ is $\leq^{\prime}$-closed.}
\\
\\
The {\em saturated model of $s_{0}$}\/ is the model $(W^{\prime},{\leq^{\prime}},{R^{\prime}},V^{\prime})$.
\\
\\
\begin{lemma}[Saturated Truth Lemma]\label{lemma:truth:lemma:at:the:end:of:the:procedure}
Let $B{\in}\Fo$.
For all $(l,v,\delta,T){\in}W^{\prime}$, if $B{\in}\Sigma_{A}$ then $(l,v,\delta,T){\models}B$ if and only if $v{\models}B$.\footnote{Here, when we write ``$(l,v,\delta,T){\models}B$'' and ``$v{\models}B$'', we mean ``$(W^{\prime},{\leq^{\prime}},{R^{\prime}},V^{\prime}),(l,v,\delta,T){\models}
$\linebreak$
B$'' and ``$(W_{c},{\leq_{c}},{R_{c}},V_{c}),v{\models}B$''.}
\end{lemma}
\begin{proof}
By induction on $B{\in}\Fo$.
Let $(l,v,\delta,T){\in}W^{\prime}$.
Suppose $B{\in}\Sigma_{A}$.
\\
\\
{\bf Case $B{=}p$:}
\\
\\
From left to right, suppose $(l,v,\delta,T){\models}p$.
Hence, $(l,v,\delta,T){\in}V^{\prime}(p)$.
Thus, $v{\in}
$\linebreak$
V_{c}(p)$.
Consequently, $v{\models}p$.
\\
\\
From right to left, suppose $v{\models}p$.
Hence, $v{\in}V_{c}(p)$.
Thus, $(l,v,\delta,T){\in}V^{\prime}(p)$.
Consequently, $(l,v,\delta,T){\models}p$.
\\
\\
{\bf Case $B{=}C{\rightarrow}D$:}
\\
\\
From left to right, suppose $(l,v,\delta,T){\models}C{\rightarrow}D$.
For the sake of the contradiction, suppose $v{\not\models}C{\rightarrow}D$.
For a while, suppose $v$ is maximal with respect to $C{\rightarrow}D$.
Since $v{\not\models}C{\rightarrow}D$, then by Lemma~\ref{lemma:maximal:rightarrow}, $v{\models}C$ and $v{\not\models}D$.
Hence, by induction hypothesis, $(l,v,\delta,T){\models}C$ and $(l,v,\delta,T){\not\models}D$.
Since $(l,v,\delta,T){\models}C{\rightarrow}D$, then $(l,v,\delta,T){\models}D$: a contradiction.
Thus, $v$ is not maximal with respect to $C{\rightarrow}D$.
Since $({\mathcal T},{\ll},{\triangleright})$ is $(\alpha_{f}{+}1)$-clean for maximality, there exists $(m,w,\epsilon,U){\in}{\mathcal T}$ such that $w{\not\models}C{\rightarrow}D$, $w$ is maximal with respect to $C{\rightarrow}D$ and $(l,v,\delta,T){\ll}(m,w,\epsilon,U)$.
Consequently, by Lemma~\ref{lemma:maximal:rightarrow}, $w{\models}C$ and $w{\not\models}D$.
Hence, by induction hypothesis, $(m,w,\epsilon,U){\models}C$ and $(m,w,\epsilon,U){\not\models}D$.
Since $(l,v,\delta,T){\models}C{\rightarrow}D$ and $(l,v,\delta,T){\ll}
$\linebreak$
(m,w,\epsilon,U)$, then $(m,w,\epsilon,U){\models}D$: a contradiction.
\\
\\
From right to left, suppose $v{\models}C{\rightarrow}D$.
For the sake of the contradiction, suppose $(l,v,\delta,T){\not\models}C{\rightarrow}D$.
Thus, there exists $(m,w,\epsilon,U){\in}{\mathcal T}$ such that $(l,v,\delta,T){\ll^{\star}}(m,
$\linebreak$
w,\epsilon,U)$, $(m,w,\epsilon,U){\models}C$ and $(m,w,\epsilon,U){\not\models}D$.
Consequently, by Lem\-ma~\ref{lemma:about:coherence:and:morphisms}, $v{\leq_{c}}w$.
Moreover, by induction hypothesis, $w{\models}C$ and $w{\not\models}D$.
Since $v{\models}C{\rightarrow}D$, then $w{\models}D$: a contradiction.
\\
\\
{\bf Case $B{=}{\square}C$:}
\\
\\
From left to right, suppose $(l,v,\delta,T){\models}{\square}C$.
For the sake of the contradiction, suppose $v{\not\models}{\square}C$.
We consider the following cases: $\delta{<}\alpha_{f}$; $\delta{=}\alpha_{f}$.
When $\delta{<}\alpha_{f}$, since $v{\not\models}{\square}C$ and $({\mathcal T},{\ll},{\triangleright})$ is $\alpha_{f}$-clean for accessibility, then there exists $(m,w,\epsilon,U){\in}{\mathcal T}$ such that $(l,v,\delta,T){\triangleright}(m,w,\epsilon,U)$ and $w{\not\models}C$.
Hence, $(l,v,\delta,T){R^{\prime}}
$\linebreak$
(m,w,\epsilon,U)$.
Since $(l,v,\delta,T){\models}{\square}C$, then $(m,w,\epsilon,U){\models}C$.
Thus, by induction hypothesis, $w{\models}C$: a contradiction.
When $\delta{=}\alpha_{f}$, let $(m,w,\beta_{f},U){\in}{\mathcal T}_{\beta_{f}}$ be such that $f(l,v,\delta,T){=}(m,w,\beta_{f},U)$.
Consequently, $v{\cap}\Sigma_{A}{=}w{\cap}\Sigma_{A}$.
Since $v{\not\models}{\square}C$, then $w{\not\models}{\square}C$.
Since $({\mathcal T},{\ll},{\triangleright})$ is $\alpha_{f}$-clean for accessibility, then there exists $(n,x,\beta_{f}{+}1,V){\in}{\mathcal T}$ such that $(m,w,\beta_{f},U){\triangleright}(n,x,\beta_{f}{+}1,V)$ and $x{\not\models}C$.
Since $f(l,
$\linebreak$
v,\delta,T){=}(m,w,\beta_{f},U)$, then $(l,v,\delta,T){R^{\prime}}(n,x,\beta_{f}{+}1,V)$.
Since $(l,v,\delta,T){\models}{\square}C$, then $(n,x,\beta_{f}{+}1,V){\models}C$.
Hence, by induction hypothesis, $x{\models}C$: a contradiction.
\\
\\
From right to left, suppose $v{\models}{\square}C$.
For the sake of the contradiction, suppose $(l,v,\delta,T){\not\models}{\square}C$.
Thus, there exists $(m,w,\epsilon,U){\in}{\mathcal T}$ such that $(l,v,\delta,T){R^{\prime}}(m,w,
$\linebreak$
\epsilon,U)$ and $(m,w,\epsilon,U){\not\models}C$.
Consequently, there exists $n{\geq}1$ and there exists $(k_{0},u_{0},\gamma_{0},Z_{0}),\ldots,(k_{n},u_{n},\gamma_{n},Z_{n}){\in}{\mathcal T}$ such that $(k_{0},u_{0},\gamma_{0},Z_{0}){=}(l,v,\delta,T)$, $(k_{n},
$\linebreak$
u_{n},\gamma_{n},Z_{n}){=}(m,w,\epsilon,U)$ and for all $a{\in}(n)$, either $(k_{a{-}1},u_{a{-}1},\gamma_{a{-}1},Z_{a{-}1}){\triangleright}(k_{a},
$\linebreak$
u_{a},\gamma_{a},Z_{a})$, or $(k_{a{-}1},u_{a{-}1},\gamma_{a{-}1},Z_{a{-}1}){\triangleright^{\prime}}(k_{a},u_{a},\gamma_{a},Z_{a})$.
Moreover, by induction hypothesis, $w{\not\models}C$.
By induction on $n{\geq}1$, the reader may easily verify that $w{\models}C$: a contradiction.
\\
\\
{\bf Case $B{=}{\lozenge}C$:}
\\
\\
From left to right, suppose $(l,v,\delta,T){\models}{\lozenge}C$.
For the sake of the contradiction, suppose $v{\not\models}{\lozenge}C$.
Since $(l,v,\delta,T){\models}{\lozenge}C$, then there exists $(m,w,\epsilon,U){\in}{\mathcal T}$ such that $(l,v,\delta,T){R^{\prime}}(m,w,\epsilon,U)$ and $(m,w,\epsilon,U){\models}C$.
Hence, there exists $n{\geq}1$ and there exists $(k_{0},u_{0},\gamma_{0},Z_{0}),\ldots,(k_{n},u_{n},\gamma_{n},Z_{n}){\in}{\mathcal T}$ such that $(k_{0},u_{0},\gamma_{0},Z_{0}){=}(l,
$\linebreak$
v,\delta,T)$, $(k_{n},u_{n},\gamma_{n},Z_{n}){=}(m,w,\epsilon,U)$ and for all $a{\in}(n)$, either $(k_{a{-}1},u_{a{-}1},\gamma_{a{-}1},
$\linebreak$
Z_{a{-}1}){\triangleright}(k_{a},u_{a},\gamma_{a},Z_{a})$, or $(k_{a{-}1},u_{a{-}1},\gamma_{a{-}1},Z_{a{-}1}){\triangleright^{\prime}}(k_{a},u_{a},\gamma_{a},Z_{a})$.
Moreover, by induction hypothesis, $w{\models}C$.
By induction on $n{\geq}1$, the reader may easily verify that $w{\not\models}C$: a contradiction.
\\
\\
From right to left, suppose $v{\models}{\lozenge}C$.
For the sake of the contradiction, suppose $(l,v,\delta,T){\not\models}{\lozenge}C$.
We consider the following cases: $\delta{<}\alpha_{f}$; $\delta{=}\alpha_{f}$.
When $\delta{<}\alpha_{f}$, since $v{\models}{\lozenge}C$ and $({\mathcal T},{\ll},{\triangleright})$ is $\alpha_{f}$-clean for accessibility, then there exists $(m,w,\epsilon,U){\in}{\mathcal T}$ such that $(l,v,\delta,T){\triangleright}(m,w,\epsilon,U)$ and $w{\models}C$.
Thus, $(l,v,\delta,T){R^{\prime}}
$\linebreak$
(m,w,\epsilon,U)$.
Since $(l,v,\delta,T){\not\models}{\lozenge}C$, then $(m,w,\epsilon,U){\not\models}C$.
Consequently, by induction hypothesis, $w{\not\models}C$: a contradiction.
When $\delta{=}\alpha_{f}$, let $(m,w,\beta_{f},U){\in}{\mathcal T}_{\beta_{f}}$ be such that $f(l,v,\delta,T){=}(m,w,\beta_{f},U)$.
Hence, $v{\cap}\Sigma_{A}{=}w{\cap}\Sigma_{A}$.
Since $v{\models}{\lozenge}C$, then $w{\models}{\lozenge}C$.
Since $({\mathcal T},{\ll},{\triangleright})$ is $\alpha_{f}$-clean for accessibility, then there exists $(n,x,\beta_{f}{+}1,V){\in}{\mathcal T}$ such that $(m,w,\beta_{f},U){\triangleright}(n,x,\beta_{f}{+}1,V)$ and $x{\models}C$.
Since $f(l,
$\linebreak$
v,\delta,T){=}(m,w,\beta_{f},U)$, then $(l,v,\delta,T){R^{\prime}}(n,x,\beta_{f}{+}1,V)$.
Since $(l,v,\delta,T){\not\models}{\lozenge}C$, then $(n,x,\beta_{f}{+}1,V){\not\models}C$.
Thus, by induction hypothesis, $x{\not\models}C$: a contradiction.
\medskip
\end{proof}
All in all, the desired result is within reach.
\begin{proposition}[Completeness and Finite Frame Property]\label{proposition:fmp:fik}
For all $A{\in}
$\linebreak$
\Fo$, the following conditions are equivalent:
\begin{enumerate}
\item $A{\in}\LIK4$,
\item for all transitive, downward confluent and forward confluent frames $(W,{\leq},
$\linebreak$
{R})$, $(W,{\leq},{R}){\models}A$,
\item for all finite, transitive, downward confluent and forward confluent frames $(W,{\leq},{R})$, $(W,{\leq},{R}){\models}A$.
\end{enumerate}
\end{proposition}
\begin{proof}
Let $A{\in}\Fo$.
\\
\\
$\bf{(1){\Rightarrow}(2)}$~By Lemma~\ref{proposition:soundness:of:the:logics}.
\\
\\
$\bf{(3){\Rightarrow}(1)}$ Suppose for all finite, transitive, downward confluent and forward confluent frames $(W,{\leq},{R})$, $(W,{\leq},{R}){\models}A$.
For the sake of the contradiction, suppose $A{\not\in}\LIK4$.
Hence, by Lemma~\ref{lemma:almost:completeness}, there exists a prime theory $s_{0}$ such that $A{\not\in}s_{0}$.
Thus, by Lemma~\ref{lemma:truth:lemma}, $s_{0}{\not\models}A$.
Let $(W^{\prime},{\leq^{\prime}},{R^{\prime}})$ be the saturated frame of $s_{0}$.
Since for all finite, transitive, downward confluent and forward confluent frames $(W,{\leq},{R})$, $(W,{\leq},{R}){\models}A$, then by Lemma~\ref{lemma:frame:prime:is:appropriate}, $(W^{\prime},{\leq^{\prime}},{R^{\prime}}){\models}A$.
Let $V^{\prime}$: $\At\longrightarrow\wp(W^{\prime})$ be the saturated valuation of $s_{0}$.
Since $(W^{\prime},{\leq^{\prime}},{R^{\prime}}){\models}A$, then $(W^{\prime},{\leq^{\prime}},{R^{\prime}},V^{\prime}){\models}A$.
Consequently, $(0,s_{0},0,0){\models}A$.
Hence, by Lemma~\ref{lemma:truth:lemma:at:the:end:of:the:procedure}, $s_{0}{\models}A$: a contradiction.
\medskip
\end{proof}
\begin{proposition}
The membership problem in $\LIK4$ is decidable.
\end{proposition}
\begin{proof}
By~\cite[Theorem~$6.13$]{Blackburn:et:al:2001} and Proposition~\ref{proposition:fmp:fik}.
\medskip
\end{proof}
%
%
%
%
%\section{Conclusion}\label{section:conclusion}
%
%
%zzzzz
%
%
%
\section*{Acknowledgements}
We wish to thank Han Gao (Aix-Marseille University), Zhe Lin (Xiamen University) and Nicola Olivetti (Aix-Marseille University) for their valuable remarks.
\\
\\
Special acknowledgement is also granted to our colleagues of the Toulouse Institute of Computer Science Research for many stimulating discussions about the subject of this article.
%
%
%\\
%\\
%
%
%We make a point of thanking as well the referees for their feedback: their useful suggestions have been essential for improving the readability of a preliminary version of this article.
%
%
%
%
\bibliographystyle{named}
\end{document}